\documentclass[a4paper,UKenglish,cleveref, autoref, thm-restate]{lipics-v2021}
\pdfoutput=1 %uncomment to ensure pdflatex processing (mandatatory e.g. to submit to arXiv)
\usepackage{mathtools}
\usepackage{stmaryrd}
\usepackage{xspace}
\usepackage[table]{xcolor}
\usepackage{colortbl}       % Ensures full-row background coloring 
\usepackage{mathpartir}     % \inferrule
\usepackage{booktabs}
\usepackage{xspace}
\usepackage{array}          % For table column customization
\usepackage{pifont}
\definecolor{istablue}{RGB}{37,59,144}
\definecolor{darkgreen}{RGB}{0,100,53}    
\definecolor{istagreen}{RGB}{0,100,53}    
\definecolor{istalightgreen}{RGB}{110,193,108} % should have 25% alpha    
\definecolor{istalightgreen2}{RGB}{113,187,111}    
\definecolor{istalightazure}{RGB}{218,238,239}    
\definecolor{istablack}{RGB}{26,26,24}       
\definecolor{istablue}{RGB}{37,59,144}    
\definecolor{istaazure}{RGB}{0,154,163}    
\definecolor{istalightorange}{RGB}{254,218,163}    
\definecolor{istalightorange2}{RGB}{252,215,184}    
\definecolor{istared}{RGB}{225,9,44}    
\definecolor{istaorange}{RGB}{247,166,0}    
\definecolor{istaorange2}{RGB}{244,152,0}

\definecolor{ana}{RGB}{0,100,200}   
\definecolor{mar}{RGB}{50,100,0}   
\definecolor{all}{RGB}{200,100,0}

\usepackage[noend,linesnumbered,commentsnumbered,ruled]{algorithm2e}

\SetCommentSty{decmt}    
\SetKwInput{KwData}{Input}    
\SetKwInput{KwResult}{Output}    
\SetKw{KwBreak}{break}    
\SetKw{KwCont}{continue}    
\SetKw{KwContinue}{continue}    
\SetNlSty{textbf}{\color{black}}{}    
\SetNlSkip{1.1em}  

\allowdisplaybreaks

\usepackage{tcolorbox}

\newcommand{\nb}[1]{\marginpar{\scriptsize #1}}

\newcommand\xxx[1]{\nb{\textcolor{istared}{#1}}}
\renewcommand{\xxx}[1]{}

\definecolor{darkgreen}{rgb}{0.0, 0.6, 0.13}
\newcommand{\diff}[2]{{\textcolor{red}{\sout{#1}}{\textcolor{darkgreen}{#2}}}}
\renewcommand{\diff}[2]{{#2}}

\newcommand{\del}[1]{\diff{#1}{}}

\newcounter{notes}

\usepackage{float}

\usepackage{tikz}
\usetikzlibrary{matrix}
\usetikzlibrary{arrows}
\usetikzlibrary{shapes}
\usetikzlibrary{automata, positioning,through,calc}

\tikzset{
->, % makes the edges directed
>=stealth, % makes the arrow heads bold
node distance=2cm, % specifies the minimum distance between two nodes. Change if necessary.
every state/.style={thick, fill=gray!10,ellipse}, % sets the properties for each ’state’ node
interm state/.style={thick, fill=gray!10, draw, rectangle, inner sep=6}, % sets the properties for each ’state’ node
initial text=$ $, % sets the text that appears on the start arrow
}

\usepackage{mdframed}
\definecolor{light-gray}{gray}{0.925}
\global\mdfdefinestyle{example}{
	linecolor=light-gray,
	backgroundcolor=light-gray,
	innerleftmargin=15,innerrightmargin=15,
	innertopmargin=15,innerbottommargin=15
}

\newcommand{\HLTL}{\mathsf{HyperLTL}}
\newcommand{\HTwoLTL}{\mathsf{Hyper}^2\mathsf{LTL}}
\newcommand{\HQPTL}{\mathsf{HyperQPTL}}
\newcommand{\LTL}{\mathsf{LTL}}
\newcommand{\QPTL}{\mathsf{QPTL}}
\newcommand{\TeamLTL}{\mathsf{TeamLTL}}
\newcommand{\logic}{\TFOL}

\newcommand{\generatedSet}[1]{\llbracket #1 \rrbracket}

\newcommand{\Quant}{\mathbb{Q}}
\newcommand{\QExists}{\mathbb{E}}

\newcommand{\setTraces}{\text{T}}
\newcommand{\allSetTraces}{\mathbb{T}}

\newcommand{\setsetTraces}{\textbf{T}}

\newcommand{\traceVar}{\pi}

\newcommand{\traceAssign}{\Pi}
\newcommand{\emptyAssign}{\traceAssign^{\emptyset}}
\newcommand{\flatT}[1]{\langle  #1 \rangle}

\newcommand{\Prop}{\mathcal{X}}
\newcommand{\PropY}{\mathcal{Y}}
\newcommand{\trace}{\tau}

\newcommand{\val}{v}

\newcommand{\Bool}{\mathbf{2}}
\newcommand{\AllVals}{\Bool^{\Prop}}

\newcommand{\Until}{\mathbin\mathbf{U}}
\newcommand{\Next}{\mathbf{X}}

\newcommand{\VarTrace}{{\Var_{\traceS}}}
\newcommand{\VarTime}{\Var_{\nat}}
\newcommand{\FOL}{\text{FO}}

\newcommand{\freeV}{\text{free}}

\newcommand{\FOLOrder}{{\FOL[<]}\xspace}

\newcommand{\timeS}{\nat_{<}}
\newcommand{\traceS}{\allSetTraces}

\newcommand{\UnaryP}{X}
\newcommand{\BinaryP}{X}
\newcommand{\TrPred}{T}

\newcommand{\Var}{\mathcal{V}}

\newcommand{\constrQ}[2]{#1 #2\!\mathbin{::}\!\TrPred}

\newcommand{\SOneS}{\texttt{S1S}\xspace}
\newcommand{\toMSO}{\texttt{toS1S}}
\newcommand{\toTime}{\texttt{toSuppSet}}
\newcommand{\toBool}{\texttt{toBool}}

\newcommand{\toHyper}{\texttt{toHyper}}

\newcommand{\AssignFirst}{\traceAssign_1}
\newcommand{\AssignSecond}{\traceAssign_2}
\newcommand{\AssignT}{\traceAssign_{\traceS}}
\newcommand{\AssignN}{\traceAssign_{\nat}}

\newcommand{\Succ}{\text{Succ}}

\newcommand{\TFOL}{{\FOL[<,\allSetTraces]}}

\newcommand{\TraceTFOL}{\setsetTraces\text{-}\TFOL}
\newcommand{\cTraceTFOL}{\setsetTraces\text{-}{\FOL[<,\constrQ{}{\allSetTraces}]}}

\newcommand{\timeVar}{i}

\newcommand{\modelsLTL}{\mathop{\models_{\texttt{\scriptsize LTL}}}}
\newcommand{\notmodelsLTL}{\mathop{\not\models_{\texttt{\scriptsize LTL}}}}

\newcommand{\modelsQPTL}{\mathop{\models_{\texttt{\scriptsize QPTL}}}}

\newcommand{\modelsHyper}{\mathop{\models_{\texttt{\scriptsize HQ}}}}
\newcommand{\notmodelsHyper}{\mathop{\not\models_{\texttt{\scriptsize HQ}}}}

\newcommand{\modelsTFOL}{\mathop{\models_{ \allSetTraces}}}
\newcommand{\notmodelsTFOL}{\mathop{\not\models_{ \allSetTraces}}}

\newcommand{\modelsSOneS}{\mathop{\models_{\texttt{\scriptsize S1S}}}}
\newcommand{\notmodelsSOneS}{\mathop{\not\models_{\texttt{\scriptsize S1S}}}}

\newcommand{\tAnd}{\text{ and }}
\newcommand{\tOr}{\text{ or }}

\newcommand{\tIf}{\text{ if }}

\newcommand{\tIff}{\text{ iff }}
\newcommand{\tSt}{\text{ s.t.\ }}

\newcommand{\Def}{\stackrel{\mathclap{\text{def}}}{=}}
\newcommand{\nat}{\mathbb{N}}

\newcommand{\LLL}{\mathord{\ldots}}

\newcommand{\As}{\mathop{:}}

\title{Flavors of Quantifiers in Hyperlogics}

\author{Marek Chalupa}{IST Austria, Klosterneuburg, Austria}{Marek.Chalupa@ist.ac.at}{https://orcid.org/0000-0003-1132-5516}{}%TODO mandatory, please use full name; only 1 author per \author macro; first two parameters are mandatory, other parameters can be empty. Please provide at least the name of the affiliation and the country. The full address is optional. Use additional curly braces to indicate the correct name splitting when the last name consists of multiple name parts.

\author{Thomas A. Henzinger}{IST Austria, Klosterneuburg, Austria}{tah@ist.ac.at}{https://orcid.org/0000-0002-2985-7724}{}

\author{Ana {Oliveira da Costa}}{IST Austria, Klosterneuburg, Austria}{ana.costa@ist.ac.at}{https://orcid.org/0000-0002-8741-5799}{}

\authorrunning{M. Chalupa, T. A. Henzinger and A. Oliveira da Costa} %TODO mandatory. First: Use abbreviated first/middle names. Second (only in severe cases): Use first author plus 'et al.'

\Copyright{ } %TODO mandatory, please use full first names. LIPIcs license is "CC-BY";  http://creativecommons.org/licenses/by/3.0/

 %TODO mandatory: Please choose ACM 2012 classifications from https://dl.acm.org/ccs/ccs_flat.cfm 

 %TODO mandatory; please add comma-separated list of keywords

\relatedversion{} %optional, e.g. full version hosted on arXiv, HAL, or other respository/website
\acknowledgements{This work was supported in part by the Austrian Science Fund (FWF) SFB project SpyCoDe 10.55776/F85 and by the ERC Advanced Grant VAMOS 101020093.  }%optional

\nolinenumbers %uncomment to disable line numbering

\EventEditors{}
\EventNoEds{0}
\EventLongTitle{45th IARCS Annual Conference on Foundations of Software Technology and Theoretical Computer Science (FSTTCS 2025)}
\EventShortTitle{FSTTCS 2025}
\EventAcronym{FSTTCS}
\EventYear{2025}
\EventDate{}
\EventLocation{}
\EventLogo{}
\SeriesVolume{}
\ArticleNo{}
\begin{document}

\maketitle

\begin{abstract}
\emph{Hypertrace logic} is a sorted first-order logic with separate sorts for time and execution traces.
Its formulas specify hyperproperties, which are properties relating multiple traces.
In this work, we extend hypertrace logic by introducing trace quantifiers that range over the set of all possible traces.
In this extended logic, formulas can quantify over two kinds of trace variables: \emph{constrained trace variables}, which range over a fixed set of traces defined by the model, and \emph{unconstrained trace variables}, which can be assigned to any trace.
In comparison, hyperlogics such as HyperLTL have only constrained trace quantifiers.
We use hypertrace logic to study how different quantifier patterns affect the decidability of the satisfiability problem.
We prove that hypertrace logic without constrained trace quantifiers is equivalent to monadic second-order logic of one successor (\(\SOneS\)), and therefore satisfiable, and that the trace-prefixed fragment (all trace quantifiers precede all time quantifiers) is equivalent to \(\HQPTL\).
Moreover, we show that all hypertrace formulas where the only alternation between constrained trace quantifiers is from an existential to a universal quantifier are equisatisfiable to formulas without constraints on their trace variables and, therefore, decidable as well.
Our framework allows us to study also time-prefixed hyperlogics, for which we provide new decidability and undecidability results.
%
%Thus, by looking through the lens of classical first-order logic, where all types of quantification are made explicit, it becomes clear that quantifier alternation over constrained trace variables is the primary source of undecidability for temporal hyperlogics.
%
\end{abstract}

% {\bf Title ideas}:
% \begin{itemize}
%     %
%     \item Beyond System Constraints: Mixed Trace Quantification in Hypertrace Logic
%     %
%     %
%     \item Quantifying Beyond the System: Hyperproperty Reasoning with Constrained and Unconstrained Trace Quantifiers
%     %
%     \item Unleashing the Quantifiers: Hyperproperties over Constrained and Unconstrained Trace Quantifiers
%     \item 
% \end{itemize}

\section{Introduction}
\label{sec:intro}

Temporal logics offer a powerful formalism for specifying \emph{trace properties}, which describe sets of allowed sequences of events (or system states)  that a system can exhibit over time.
Within the linear-time spectrum, Linear Temporal Logic (\(\LTL\)) has been particularly successful in system verification, striking a practical balance between expressiveness and decidability. 
Its extension with propositional quantifiers, known as Quantified Propositional Temporal Logic (\(\QPTL\)), extends \(\LTL\)'s expressive power to capture all regular trace properties while preserving decidability.
The connection between temporal logics and classical first-order logic for specifying linear-time properties was first explored in the seminal work of Kamp~\cite{Kamp1968}. 
This line of research established that \(\LTL\) is expressively equivalent to first-order logic of order (\(\FOLOrder\)), where the order relation \(<\) is interpreted over the natural numbers.
In the case of \(\QPTL\), it was proven to be expressively equivalent to monadic second-order logic of order of one successor (\SOneS) \cite{qptl21202}.
However, these logics are inherently limited to reasoning about individual execution traces and cannot express properties that involve comparing multiple executions. 
In particular, they are not capable of capturing \emph{hyperproperties}.

To address this limitation of trace-based formalisms in expressing hyperproperties, numerous extensions have been proposed in the literature.
Notably, \(\HLTL\) and \(\HQPTL\) extend \(\LTL\) and \(\QPTL\), respectively, with quantification over traces of the system under verification.
An alternative approach involves extending \(\FOLOrder\) with capabilities to compare multiple traces.
\emph{Hypertrace logic}, introduced in \cite{bartocci2022flavors}, extends \(\FOLOrder\) into a two-sorted first-order logic, with separate sorts for time and traces of a given system, allowing only binary predicates over traces and time.
%
%Being a two-sorted logic, hypertrace logic effectively supports \del{for} two types of quantifiers: time quantifiers over the natural numbers, and trace quantifiers \del{s}over the set of traces provided by the model.

In this work, we propose an extension of hypertrace logic that includes both \emph{constrained trace quantifiers}, ranging over all traces in the set given as a model, and \emph{unconstrained trace quantifiers}, ranging over the universe of all traces.
Properties with a mix of constrained and unconstrained trace quantification occur naturally in many areas of system verification.
For example, in reactive systems,  \emph{input enableness} requires that \emph{``the system must produce an output for any possible input''} which can be specified as:
\begin{align}
\label{intro:input_enable}
    \forall \traceVar\ \constrQ{\exists}{\traceVar'}\ \forall i\ \big((\mathrm{input}(\traceVar,i){\mathop{\leftrightarrow}}\mathrm{input}(\traceVar',i)) \wedge \mathrm{outputs}(\traceVar',i)\big).
\end{align}
By using the unconstrained quantifier,  \(\forall \traceVar\), we ensure that all possible input values are taken into account, not just those accepted by the reactive system.
By linking {the universally quantified} inputs to those instantiated by the constrained existential quantifier, \(\constrQ{\exists}{\traceVar'}\), we guarantee that {every possible} input {is} in at least one system execution.

\newcommand{\LIN}{linearizability\xspace}
\newcommand{\linear}{\textit{seq}}
\newcommand{\EvPrec}{\textit{Order}}
Another example, can be found in the specification of consistency properties in concurrent systems like
\emph{linearizability} \cite{LinearHerlihyWing90}, which requires all histories of invocations and response events to shared resources to be consistent with some sequential execution of those operations. 
While verifying for \LIN, it is irrelevant whether the implementation of a concurrent data structure exhibits linear histories because all that matters is that its histories can be successfully related to the ``idealized'' sequential implementation.
In its essence, linearizability can be captured by the following hyperproperty:
\begin{align}
\constrQ{\forall}{\traceVar}\  \exists \traceVar' \ (\varphi_{\text{linear}}(\traceVar') \wedge    
  \varphi_{\text{equivOrder}}(\traceVar,\traceVar')).
\label{def:intro:linear}
\end{align}
The formula above universally quantifies over the observed call history of the shared object while requiring the existence of a linear trace (\(\varphi_{\text{linear}}(\traceVar')\)) that is equivalent to the observed one and respects the precedence order defined between the observed (\(\varphi_{\text{equivOrder}}(\traceVar,\traceVar')\)).

%\medskip
We prove that the fragment of hypertrace logic without constrained trace quantifiers, named \emph{unconstrained hypertrace logic}, is expressively equivalent to \(\SOneS\).
It follows directly that unconstrained hypertrace logic has a decidable satisfiability problem.
We show how to translate hypertrace formulas whose only constrained trace quantifier alternation is from existential to universal quantification into an equisatisfiable unconstrained hypertrace formula.
Hence, this fragment also has a decidable satisfiability problem.
This fragment places no limitations on the positioning of unconstrained or temporal quantifiers after the constrained quantification.
For the fragment where all trace quantifiers occur before time quantifiers, we prove it to be expressively equivalent to \(\HQPTL\).
While the majority of satisfiability results for hyperlogics in the literature focus on the trace-prefix fragment, we extend this line of work by also analyzing the fragment of hypertrace logic where time quantifiers occur first.
For the time-prefix fragment, we prove a new undecidability result by presenting a reduction from the non-halting problem for 2-counter Minsky machines.
Our main contribution is in highlighting the role of different quantifiers, including unconstrained trace quantification, plays in defining hyperlogics with a decidable satisfiability problem.

\vspace{-2mm}
\section{Logics of Time and Order}
\label{sec:pre}
\vspace{-2mm}
In this section, we present the necessary theoretical background, including definitions and prior results, which will form the basis for the results introduced in the next sections.

%\paragraph*{Trace Properties.}
Let \(\Prop\) be a finite set of boolean variables.
A \emph{valuation} is a partial mapping assigning boolean values, \(\Bool\mathop{=}\{0,1\}\), to variables in \(\Prop\), that is,  \(\val\As \Prop \rightarrow \Bool\).
We denote by \(\val[x \mapsto b]\) the valuation resulting from
updating the value of \(x\) in \(\val\) to \(b\),  and
the \emph{set of all valuations of \(\Prop\)}  by \(\AllVals\).
We freely treat a valuation as a set of variables where suitable.
In particular, we write $x\mathop{\in}\val$ when $\val(x)\mathop{=}1$.
%and we can list a valuation extensionally as $\{x \mid \val(x) \} \cup \{\neg x\mid \neg\val(x)\}$. %
%For example, $\{x, \neg y\}$ is the valuation $\{x \mapsto 1, y\mapsto 0\}$.
%
A \emph{trace \(\trace\) over variables \(\Prop\)}  is a sequence of valuations in \(\AllVals\).
The set of all \emph{finite} traces over \(\Prop\)  is denoted  \((\AllVals)^*\),
while the set of all
\emph{infinite} traces over \(\Prop\) is denoted by \((\AllVals)^\omega\).
%
%We denote the set with all finite and infinite traces as \((\AllVals)^{\fininf} = (\AllVals)^* \mathop{\cup} (\AllVals)^{\omega}\).
%
%
%The \emph{length} of a finite trace \(\trace = v_0 v_1 \dots v_n\) is defined as \(|\trace| = n + 1\), while for an infinite traces \(|\trace| = \omega\).
%
%
For a trace \(\trace = v_0 v_1 \dots\) and an index \(i\) within its length (i.e., \(i < |\trace|\)), we adopt the following indexing notations:
\(\trace[i] = v_i\), \(\trace[i\ldots] = v_i v_{i+1}\dots\), and \(\trace[\ldots i] = v_0 v_1\dots v_{i-1}\).
When an index falls outside a trace length (i.e., \(j \geq |\trace|\)), we adopt the convention that
\(\trace[j\ldots]\) is the empty trace, and \(\trace[\ldots j] = \trace\).
Two traces, \(\trace\) and \(\trace'\), agree on their valuation of a set of propositional variables \(\PropY\) at position \(i\), denoted \(\trace[i]\mathop{=}_{\PropY} \trace'[i]\), iff
\((\trace[i]\mathop{\cap}\PropY) \mathop{=}(\trace'[i]\mathop{\cap}\PropY)\).
We generalize it to equality between two traces of equal size with respect to a given set of variables \(\PropY\), denoted \(\trace\mathop{=}_{{\PropY}}\trace'\), when, for all \(i\mathop{<}|\trace|\), \(\trace[i]\mathop{=}_{\PropY}\trace'[i]\).

\emph{Trace properties} define set of traces satisfying a target specification. 
Formally, a trace property \(\setTraces\) over \(\Prop\)  is a set of traces over \(\Prop\); for infinite traces,
that is, \(\setTraces \mathop{\subseteq} (\AllVals)^\omega\). 
A trace \(\trace\) satisfies a trace property \(\setTraces\) iff it is one of its elements, that is \(\trace\mathop{\in} \setTraces\).
We refer to the set of all trace properties as \(\allSetTraces \mathop{=} \Bool^{(\AllVals)^\omega}\).
For systems represented by a set of traces \(S\), where each trace corresponds to one of the system's execution, the system is said to satisfy the trace property \(\setTraces\) if and only if all its traces are in \(\setTraces\); that is,  \(S\mathop{\subseteq} \setTraces\).
A \emph{hyperproperty} \(\setsetTraces\) defines a set of systems (or, equivalently a set of trace properties) satisfying a given specification, \(\setsetTraces\mathop{\subseteq}\allSetTraces\).
Hence, a system \(S\) satisfies a hyperproperty \(\setsetTraces\) if and only if \(S\mathop{\in}\setsetTraces\).

\vspace{-2mm}
\subsection{Linear Temporal Logics}
\label{sec:LTL}
\vspace{-2mm}

A successful formalism to specify trace properties is Linear Temporal Logic (\(\LTL\)), introduced by Pnueli in \cite{Pnueli77}.
LTL formulas are defined by the grammar:
\(\varphi ::= \ a \, |\, \neg \varphi \, |\, \varphi \vee \varphi \, | \, \Next \varphi \, | \, \varphi \Until \varphi\)
where \(a\mathop{\in} \Prop\) 
%\xxx{Sometimes $a$ is a propositional variable, sometimes $x$. Maybe we should choose one (a,b,c,... or x,y,z,...)}
%\xxx{\new{Ana: Let's use a,b,c... and q, p,....}.}
is a propositional variable, and \(\Next\) (``next'') and \(\Until\) (``until'') are temporal modalities.
%
%LTL formulas are interpreted over infinite traces and boolean domains.
%
Given a LTL formula \(\varphi\) and an infinite trace \(\trace \mathop{\in} (\Bool^{\Prop})^\omega\), the trace \(\trace\) satisfies \(\varphi\), denoted \(\trace \modelsLTL \varphi\), inductively over the structure of \(\varphi\), as follows:
\[
\begin{split}
	&\trace \modelsLTL a \tIff a\mathop{\in}\trace[0] \hspace{57.5mm}
	\trace \modelsLTL \neg  \psi \tIff \trace \notmodelsLTL  \psi \\
	&  \trace \modelsLTL  \psi_1 \vee  \psi_2 \tIff 
	\trace \modelsLTL  \psi_1 \tOr \trace \modelsLTL  \psi_2 \hspace{25mm}
	\trace \modelsLTL \Next  \psi \tIff  \trace[1\ldots] \modelsLTL  \psi\\
	&\trace \modelsLTL  \psi_1\! \Until\!  \psi_2 \tIff 
	\text{exists } j\mathop{\geq} 0 \tSt \trace[j\ldots] \modelsLTL  \psi_2 \tAnd
	\text{for all } 0\mathop{\leq} j'\mathop{<}j, \trace[j'\ldots] \modelsLTL  \psi_1.
\end{split}
\]
%
%When clear from the context, we may omit the \(\mathrm{LTL}\) subscript from \(\modelsLTL\).
%
\(\QPTL\) extends \(\LTL\) with propositional quantification.
Concretely, \(\QPTL\) formulas may include subformulas with propositional quantifiers (e.g., \(\exists q\ \varphi\), where \(q\mathop{{\in}}\Prop\) and \(\varphi\) is a \(\QPTL\) formula) interpreted as:
\(
\trace \modelsQPTL \exists q\ \psi \tIff \text{ exists } \trace'
%\mathop{\in}(\Bool^\Prop)^{\omega} 
\tSt
\trace\mathop{=}_{\Prop\mathop{\setminus}\{q\}} \trace'
%\text{ s.t.\ for all } i\mathop{\in}\nat\ (\trace'[i]\mathop{\cap}(\Prop\mathop{\setminus}\{q\}))\mathop{=}(\trace[i]\mathop{\cap}(\Prop\mathop{\setminus}\{q\})) 
\tAnd \trace' \modelsQPTL  \psi.
\)
All other subformulas are interpreted as for \(\LTL\).

\vspace{-2mm}
\subsection{Monadic Logics of Order}
\label{sec:S1S}
\vspace{-2mm}

Monadic second-order logic of one successor (\(\SOneS\)) is a widely used formalism for reasoning about regular properties of infinite sequences.
S1S formulas \(\varphi\) are defined by the grammar:
\begin{align}
\varphi ::= \exists \UnaryP\, \varphi \ |\  \exists \timeVar\, \varphi \ |\  \neg \varphi \ |\ \varphi \vee \varphi \ |\ \timeVar = \timeVar \ | \ \Succ(\timeVar, \timeVar) \ |\ \UnaryP(\timeVar)
\label{def:secondOrder}
\end{align}
where \(\UnaryP\) is a second-order variable from a set of variables \(\Var_2\), \(\timeVar\) is a first-order variable over a set of variables \(\Var_1\) and 
\(\Succ\) is a successor function.
The set of first and second order variables are disjoint; that is, 
\(\Var_1 \mathop{\cap} \Var_2 \mathop{=} \emptyset\).
Unless specified otherwise, we use uppercase letters for second-order variables, 
\(\Var_2\mathop{=}\{X, Y, \ldots \}\), and lowercase letter for first-order variables, 
\(\Var_1\mathop{=}\{i, j, \ldots\}\).
The set of free variables of a formula (i.e., not bounded to a quantifier) and closed formulas (no free variables) are defined as usual.

We interpret \(\SOneS\) formulas over the natural numbers, where the \(\Succ(n,n')\) is interpreted as usual: \(\Succ(n,n')\) iff \(n'\mathop{=}n+1\).
We adopt the following abbreviations:
\begin{itemize}
    \item \(X\mathop{\subseteq} Y \ \Def\  \forall i\ (X(i) \rightarrow Y(i))\);
    \item \(\mathrm{SuccClosed}(X)\ \Def\ \forall i\ \forall i'\ ((X(i) \wedge \Succ(i,i')) \rightarrow X(i')) \);
    \item \(x\mathop{\leq}y\  \Def\ \forall Z ((Z(x) \wedge \mathrm{SuccClosed}(Z)) \rightarrow Z(y))\);
    \item \(x\mathop{<}y\  \Def\ x\mathop{\leq}y \land \neg(x = y)\);
    \item \(i = 0 \ \Def\ \forall j\ (i=j \vee i < j) \).
\end{itemize}
The semantics of \SOneS formulas is defined inductively with respect to two assignments: one for first-order variables, \(\AssignFirst\As \Var_1 \mapsto \nat \), and another for second-order variables, \(\AssignSecond\As \Var_2 \mapsto \Bool^{\nat}\). 
For any assignment \(\traceAssign\),  \(\traceAssign[x\mapsto v]\) denotes the assignment that is the same as \(\traceAssign\) except for the value of \(x\) that is updated to \(v\).
We define that \((\AssignFirst, \AssignSecond)\) satisfies \(\varphi\) inductively, as follows:
\begin{align*}
% Second-order Quant
&(\AssignFirst, \AssignSecond) \modelsSOneS \exists X\ \varphi \tIff
\text{exists } S\mathop{\subseteq}\nat \tSt (\AssignFirst, \AssignSecond[X\mapsto S]) \modelsSOneS \varphi\\
% First-order Quant
&(\AssignFirst, \AssignSecond) \modelsSOneS \exists i\ \varphi \tIff
\text{exists } k\mathop{\in}\nat \tSt (\AssignFirst[i\mapsto k], \AssignSecond) \modelsSOneS \varphi\\
% Boolean
&(\AssignFirst, \AssignSecond) \modelsSOneS \neg \varphi \tIff 
(\AssignFirst, \AssignSecond) \notmodelsSOneS \varphi\\
&(\AssignFirst, \AssignSecond) \modelsSOneS  \varphi \vee \varphi' \tIff 
(\AssignFirst, \AssignSecond) \modelsSOneS \varphi \tOr (\AssignFirst, \AssignSecond) \modelsSOneS \varphi'\\
% Temporal
&(\AssignFirst, \AssignSecond) \modelsSOneS  i=j \tIff \AssignFirst(i) = \AssignFirst(j)\\
&(\AssignFirst, \AssignSecond) \modelsSOneS  \Succ(i,i') \tIff \AssignFirst(i')\mathop{=} \AssignFirst(i)+1\\
&(\AssignFirst, \AssignSecond) \modelsSOneS  \UnaryP(i) \tIff \AssignFirst(i) \mathop{\in} \AssignSecond({\UnaryP})
\end{align*}
% A pair of assignments \((\AssignFirst, \AssignSecond)\) defines the trace \(\toTraceS{\AssignFirst, \AssignSecond}\) over 
% \(\Bool^{\Var_1 \mathop{\cup} \Var_2}\) where, 
% %
% for all first-order variables \(i\mathop{\in}\Var_1\) and 
% second-order variables \(\UnaryP\mathop{\in}\Var_2\):
% \begin{align*} 
% i\mathop{\in}\toTraceS{\AssignFirst, \AssignSecond}[k] \tIff \AssignFirst(i)\mathop{=}k 
% %\\ 
% \hspace{5mm}
% \tAnd
% \hspace{5mm}
% \UnaryP\mathop{\in}\toTraceS{\AssignFirst, \AssignSecond}[k] \tIff k\mathop{\in}\AssignSecond(\UnaryP).
% \end{align*}
%\xxx{Do we need to define $\toTraceS$? Cannot $\generatedSet{\varphi}$ be a set of assignments? (See later a related comment at hypernode logic)}
%
\SOneS formulas define a set with all their satisfying assignments:
\(\generatedSet{\varphi} \mathop{=} \{(\AssignFirst, \AssignSecond)
\, |\, (\AssignFirst, \AssignSecond)\! \modelsSOneS\! \varphi \}\).

\begin{theorem}[\cite{S1SBuchi60}]
\label{thm:s1s:dec}
For all \(\SOneS\) formulas \(\varphi\),
it is decidable to determine whether \(\generatedSet{\varphi}\mathop{\neq}\emptyset\).
\end{theorem}

In his seminal work~\cite{Kamp1968}, Kamp studied of relative expressiveness between the classical first-order approach to specify linear-time -- the first-order logic of order, \(\FOLOrder\) -- and LTL. 
The first-order logic \(\FOLOrder\) is interpreted over labeled linear orders with all uninterpreted predicates being unary.
Formally, \(\FOLOrder\) formulas \(\varphi\) are defined by the grammar:
\begin{align}
\varphi ::= \exists \timeVar\, \varphi \ |\  \neg \varphi \ |\ \varphi \vee \varphi \ |\  \timeVar < \timeVar \ |\ \timeVar = \timeVar \ | \ \UnaryP(\timeVar)
\label{def:folOrder}
\end{align}
%
%\xxx{In examples later we use lowercase letters for unary predicates. Use it also in the definition of syntax?}
where \(\timeVar\) is a first-order variable, \(=\) is equality, \(<\) is the order, and \(\BinaryP\) is predicate from a given set of monadic predicates \(\Prop\).
%
%\diff{We observe that we can define the successor relation as an abbreviation:
%\(\Succ(i,i') \ \Def \ \forall j\ (i < j \rightarrow (j = i' \vee i' < j))\).}{}
%\xxx{In FO[<] defined as a restriction of S1S, we need to define <, so this sentence seems redundant.}

We work with interpretations of \(\FOLOrder\) over labeled linear-orders defined by a tuple \((\Lambda ,<, \mathcal{I})\) where
\(<\) defines a linear order over the domain \(\Lambda\), and \(\mathcal{I}\) is a function that associates to each predicate symbol a subset of elements of $\Lambda$.
From a trace \(\trace\), defined over the set of propositional variables \(\Prop\), we can derive the labeled linear-order \((\nat,<, \mathcal{I}_{\trace})\) {over} the set of unary predicates 
\(\{\UnaryP_a\ |\ a \mathop{\in} \Prop\}\), 
where \(<\) is interpreted as usual for natural numbers and 
\({\mathcal{I}_{\trace}(\UnaryP_a)=\{j\, |\, a\mathop{\in}\trace[j]\}}\). 
Using {this interpretation}, it is straightforward to translate LTL formulas to equivalent \(\FOLOrder\) formulas interpreted over \((\nat,<)\). 
Hence, \(\FOLOrder\) subsumes LTL.
The only question remaining is whether LTL subsumes \(\FOLOrder\) interpreted over \((\nat,<)\).
%, i.e., whether LTL is expressively complete for this fragment of \(\FOLOrder\).
%
Kamp proved in \cite{Kamp1968} that LTL with both past and future temporal operators is {equivalent to} 
%\xxx{'complete' undefined} for 
\(\FOLOrder\) over Dedekind
complete orders.
Later,  Gabbay et al.~\cite{gabbay1980temporal} proved that when considering only future operators, LTL is complete for \(\FOLOrder\) under the interpretation of \(<\) over the natural numbers.
For the remaining of the manuscript, we are only interested in the linear order over natural numbers.

\begin{theorem}[\cite{Kamp1968,gabbay1980temporal}]
\label{thm:ltl_FOOrder}
%\xxx{Revise this statement.}
\(\LTL\) and \(\FOLOrder\) interpreted over \((\nat,<)\) are equally expressive.
\end{theorem}

%-------------------------------------
\section{Hypertrace Logic}
\label{sec:logic}
In this section, we extend hypertrace logic, introduced in \cite{bartocci2022flavors}, to support \emph{quantification over unconstrained trace variables}; that is, variables ranging over the set of all possible traces.

\vspace{-2mm}
\subsection{Syntax and Semantics}
\vspace{-2mm}
Hypertrace logic \cite{bartocci2022flavors}, denoted  \(\TFOL\), extends $\FOLOrder$ with a time sort \(\timeS\) and a trace sort \(\traceS\).
While \(\FOLOrder\) allows only unary uninterpreted predicates\footnote{The binary predicates for equality, \(=\), and the linear order, \(<\), are interpreted by \((\nat, <)\).}, in \(\logic\) we allow binary uninterpreted predicates defined over pairs of a trace and a time variable.
In \cite{bartocci2022flavors}, hypertrace logic is defined only over \emph{constrained trace variables}: trace variables ranging over a fixed set of traces.
Here we lift the implicit constraints on trace variables: we quantify instead over the set of all possible traces while enabling trace variables to be constrained by a unary predicate.
Formally, in this work, hypertrace formulas \(\varphi\) are defined by the grammar:
\begin{align}
	\varphi &::= \exists \traceVar\ \varphi \ |\  \constrQ{\exists}{\traceVar}\ \varphi \ |\  \exists \timeVar\ \varphi \ |\ \neg \varphi \ |\ \varphi \vee \varphi  \ | \  \timeVar < \timeVar \ |\ \timeVar = \timeVar  \ |\  \BinaryP(\traceVar,\timeVar)\label{def:hypertrace:grammar:quant}
\end{align}
where \(\traceVar\) is a trace variable from a set \(\Var_{\traceS}\), 
\(\timeVar\) is a time variable from a set \(\Var_{\nat}\) disjoint from \(\Var_{\traceS}\),
\(\TrPred\) is a unary predicate over trace variables
and
\(\BinaryP\) is a binary predicate over pairs of trace and time variables from a finite set \(\mathbb{X}\).
We typically use lowercase letters for predicates in \(\mathbb{X}\) to reflect their correspondence with propositional variables in traces.
Without loss of generality, we assume that all variables are quantified only once.
We refer to \(\constrQ{\exists}{\traceVar}\) as a \emph{(existential) constrained trace quantifier} and we may refer to \(\traceVar\) as a \emph{constrained trace variable}.
We introduce the usual abbreviations:
\(\forall x\  \varphi \Def \neg (\exists x\  \neg \varphi)\), where \(x \mathop{\in}\{i, \traceVar\}\);
\(\constrQ{\forall}{\traceVar}  \varphi \Def \neg (\constrQ{\exists}{\traceVar}\ \neg \varphi)\); and
\(\varphi \wedge \varphi' \Def \neg (\neg \varphi \vee \neg \varphi')\).
Finally, constrained trace quantifiers are abbreviations for the following first-order formulas:
\(\constrQ{\exists}{\traceVar}\  \varphi \Def \exists \traceVar\  (\TrPred(\traceVar) \wedge \varphi)\) and \(\constrQ{\forall}{\traceVar}\  \varphi \Def \forall \traceVar\  (\TrPred(\traceVar) \rightarrow \varphi)\).

\newcommand{\secr}{\mathrm{secret}}
\newcommand{\pub}{\mathrm{pub}}
\begin{example}%{Example: }
\label{ex:hl}
We define independence between a secret input and a public output as:
\(\varphi\Def 
{\constrQ{\forall}{\traceVar}}\,
{\constrQ{\forall}{\traceVar'}}\, 
{\constrQ{\exists}{\traceVar_{\exists}}}\  \forall i\   
((\secr(\traceVar,i)\leftrightarrow\secr(\traceVar_{\exists},i)) \wedge
(\pub(\traceVar',i)\leftrightarrow\pub(\traceVar_{\exists},i))).
%\end{align*}
\)
% The hypertrace formula \(\varphi\) specifies that \emph{``there exists a time point \(i\) such that for all traces in a given set \(\setTraces\), the proposition \(a\) holds at that time \(i\)''}:
% %
% \(\new{\varphi_0}\ \Def\ \exists i\ \constrQ{\forall}{\traceVar}\  a(\traceVar,i).\)
% %
% An example of a formula with a different order for the trace and time quantifiers is 
% the formula \(\varphi'\) requiring that \emph{``there exist two traces in \(\setTraces\) that are each other complement with respect to the value of \(a\)''}:
% %
% \(\new{\varphi_1}\ \Def\ \constrQ{\exists}{\traceVar}\ \constrQ{\exists}{\traceVar'}\ \forall i\ a(\traceVar,i) \leftrightarrow \neg a(\traceVar',i).\)
%
We can mix constrained and unconstrained quantifiers to require that a system produces all combinations of visible public outputs with possible input secret values (i.e., a combination of input enableness with independence between secret inputs and public outputs):
%
%\begin{align*}
\(\varphi_{\text{mix}}\Def 
\forall \traceVar\
{\constrQ{\forall}{\traceVar'}}\ 
{\constrQ{\exists}{\traceVar_{\exists}}}\  \forall i\   
((\secr(\traceVar,i)\leftrightarrow\secr(\traceVar_{\exists},i)) \wedge
(\pub(\traceVar',i)\leftrightarrow\pub(\traceVar_{\exists},i))).
%\end{align*}
\)
\end{example}

%\paragraph*{Models.}
%
%We are interested in specifying hyperproperties with \(\logic\) formulas. 
%
Our models of interest are sets of traces.
As \(\logic\) is a first-order logic, we start by defining how to translate a set of traces to a first-order structure.
%
%In the translation from LTL to \(\FOLOrder\), each propositional variable \(a\mathop{\in}\Prop\) is translated to a monadic predicate \(\BinaryP_a(k)\), asserting that ``variable \(a\) is true at time \(k\mathop{\in}\timeS\)''. 
%
%Hypertrace formulas are defined with binary predicates over pairs of trace and time variables, i.e., with type \(\traceS \times \timeS\).
% 
We translate each propositional variable \(a\mathop{\in}\Prop\) to the binary predicate \(\BinaryP_a(\trace,k)\), asserting that ``variable \(a\) is true in trace \(\trace\mathop{\in}(\AllVals)^{\omega}\) at time \(k\mathop{\in}\nat\)''. 
Formally, given a set \(\setTraces\) of traces, we translate \(\setTraces\) to a structure \(\overline{\setTraces}\) with signature:
%
%\begin{align}
\((\nat, (\AllVals)^{\omega};\ \mathord<\mathop{\subseteq} \nat \times \nat, \TrPred\mathop{\subseteq} (\AllVals)^{\omega},  {(\BinaryP_a\mathop{\subseteq} (\AllVals)^{\omega} \mathop{\times} \nat)_{a \mathop{\in} \Prop}})\)
%\end{align}
%\label{def:hypertrace:traceSet_signature}
where \(\nat\) and  \((\AllVals)^{\omega}\) are the time and trace sort domains, respectively.
The predicate \(<\) is interpreted as the usual order on natural numbers,
the unary predicate {$\setTraces$} is true for all traces in the set set of traces used to generated the structure\footnote{ 
	We slightly abuse notation by letting 
	$\setTraces$
	denote both the set of traces generating the structure and the predicate determining which traces constrain the trace quantification.
	Note that, the set used to define the structure may be renamed, but the predicate 
	$\setTraces$ is true only for the traces in that set.
},
while for all variables \(a \mathop{\in} \Prop\):
\(
\BinaryP_a \mathop{=} \{(\trace, k) \ | \ a\mathop{\in}\trace[k]\}.
\)

We evaluate hypertrace formulas over pairs of assignments for traces and time:
%\begin{align}
\((\AssignT, \AssignN):
(\Var_{\traceS} \rightarrow (\AllVals)^{\omega}) \times (\Var_{\nat} \rightarrow \nat).\)
%\label{def:hypertrace:assignments:unconst}
%\end{align}
%
%We may abuse of notation, and denote by  \({(\AssignT, \AssignN)[x \mapsto v]}\) the update on the assignment matching the variable \(x\).
%
%For instance, if \(x\mathop{\in}\VarTrace\) then \({(\AssignT, \AssignN)[x \mapsto v]}\mathop{=}(\AssignT[x \mapsto v], \AssignN)\).
%\xxx{Strictly speaking, with this notation for update, the assignment should an object from $\AssignT + \AssignN$ instead of a pair}
%\xxx{\new{Ana: Is it better like this?}}
%
A set \(\setTraces\) of traces is a model of a hypertrace formula \(\varphi \mathop{\in} \TFOL\), denoted \(\setTraces \modelsTFOL \varphi\), iff \(\overline{\setTraces}\) models \(\varphi\) under the standard first-order semantics.
Formally, \(\setTraces \modelsTFOL \varphi\), iff there exists a pair of assignments \((\AssignT, \AssignN)\) such that
\((\overline{\setTraces},(\AssignT, \AssignN)) \models \varphi\), where \(\models\) is the standard first-order logic models relation.
In particular, trace quantifiers are interpreted as:
\begin{align*}
% Unconstrained
&(\overline{\setTraces},(\AssignT, \AssignN)) \modelsTFOL \exists \traceVar\ \varphi \tIff
\text{exists } \trace\mathop{\in} (\AllVals)^{\omega} \tSt (\overline{\setTraces},(\AssignT[\traceVar\mapsto\trace], \AssignN)) \modelsTFOL  \varphi\\
% Constrained
&(\overline{\setTraces},(\AssignT, \AssignN)) \modelsTFOL \constrQ{\exists}{\traceVar}\ \varphi \tIff
\text{exists } \trace\mathop{\in} \setTraces  \tSt (\overline{\setTraces},(\AssignT[\traceVar\mapsto\trace], \AssignN)) \modelsTFOL  \varphi.
\end{align*}
Hypertrace formulas define a set with all the set of traces it satisfies:
\(\generatedSet{\varphi} \mathop{=} \{\, \setTraces \ |\ \setTraces \modelsTFOL \varphi\}\).
%
% The hyperproperty defined by the hypertrace formula is the set with the union of each of those sets of traces with the trace defined 
% by translating the trace assignment into a trace over the set of free variables of the formula.
% %
% %\end{align*}
% %
% Then, a hypertrace formula \(\varphi\) defines the hyperproperty, over the set of variables \(\Prop \mathop{\cup} \Prop_{\text{free}(\varphi)}\):
% \xxx{Should $free(\varphi)$ be $\Prop_{\VarTrace}$?}
% \xxx{Define/Explain free variables.}
% \begin{align}
% \generatedSet{\varphi} = \{\, \setTraces\cup \{\toTraceH{\AssignT}\} \ |\ (\overline{\setTraces},(\AssignT, \AssignN)) \modelsTFOL \varphi\, \}.
% \end{align}
%
%
% \marek{
% I do not think this is true -- if the formula is closed, than \emph{any} assignment satisfies the formula and therefore $\{\toTraceH{\AssignT}\}$ is never empty (btw. by definition it is never empty anyway). In any case, the definition of $\generatedSet{\varphi}$ is somewhat weird.\\
% Here's a suggestion: why don't we define the set of models just as a tuple of traces and assignments? We do not use the traces anywhere (as far as I know, maybe I have overlooked something or I'm not that far in reading), we care only about if there is or is not a model.
% So we could define
% \[
% \generatedSet{\varphi} = \{\, (\setTraces, \AssignT, \AssignN) \ |\ (\overline{\setTraces},(\AssignT, \AssignN)) \modelsTFOL \varphi\, \}
% \].
% The same we could do for other logics too and it should work, because we don't care what are the models, only if there are some.}
%
From now on, we may refer to \(\BinaryP_a\) just as \(a\), and omit the subscript \(\allSetTraces\) in \(\modelsTFOL\),
when clear.

\begin{example}%{Example: }
\newcommand{\Sa}{\{a\}}
\newcommand{\Sna}{\{\ \}}
Looking back to our previous example, consider the 
sets of traces defined over the variables \(\secr\) and \(\pub\), where
with each valuation \(v\) is represented as a set:
\(\setTraces_{0}\mathop{=} \{\{\}^\omega\}\),
\(\setTraces_1\mathop{=} \{\{\secr, \pub \}^\omega\}\),  and
\(\setTraces_{0,1}\mathop{=} \{\{\}^\omega, \{\secr, \pub \}^\omega\}\del{\}}\).
For the fully constrained formula in {Example \ref{ex:hl}}, \(\setTraces_0\modelsTFOL\varphi\) and \(\setTraces_1\modelsTFOL\varphi\) because in each of these traces we only observe one of the possible values for \(\secr\). While, \(\setTraces_{0,1}\notmodelsTFOL\varphi\).
For the mix formula, \(\varphi_{\text{mix}}\) none of sets satisfy its requirements; that is,  \(\setTraces_0\notmodelsTFOL\varphi_{\text{mix}}\), \(\setTraces_1\notmodelsTFOL\varphi_{\text{mix}}\) and \(\setTraces_{0,1}\notmodelsTFOL\varphi_{\text{mix}}\).
%\(\setTraces\mathop{=} \{\Sna\Sa\Sna\Sa\Sna^\omega,  \Sna\Sna\Sna\Sa^\omega\}\) and
%\(\setTraces'\mathop{=}\{(\Sna\Sa)^\omega, (\Sa\Sna)^\omega,  \Sna^\omega\}\).
%
% {
% \[
% \begin{split}
% \setTraces = \{ \Sna\Sa\Sna\Sa\Sna^\omega,  \qquad & \setTraces' = \{ (\Sna\Sa)^\omega,\\
%  \Sna\Sna\Sna\Sa^\omega\}	\qquad \quad		  & \qquad \ \ \,  (\Sa\Sna)^\omega,\\
%                                                    & \qquad \ \ \ \, \Sna^\omega\}
% \end{split}
% \]
% }
%
% The set \(\setTraces\) satisfies the formula \(\new{\varphi_0}\) (\(\setTraces \modelsTFOL \new{\varphi_0}\)) \new{from Example~\ref{ex:hl}} because for both traces in \(\setTraces\) at time 3 the value of \(a\) is \new{true}.
% %
% While the set \(\setTraces'\) satisfies the formula \(\new{\varphi_1}\) (\(\setTraces' \modelsTFOL \new{\varphi_1}\)) with the two top most traces witnessing the satisfaction of the existential requirement of~\(\new{\varphi_1}\).
% %
% On the other hand, \(\setTraces \notmodelsTFOL \new{\varphi_1}\) and \(\setTraces' \notmodelsTFOL \new{\varphi_0}\).
%
\end{example}

\newcommand{\rewriteFlat}{\mathtt{flatten}}
\newcommand{\VCons}{\Var^c}

We close this section 
by defining the function \(\rewriteFlat\) that rewrites hypertrace formulas into an equisatisfiable formula, exploiting the independence between variable valuations in traces assigned to unconstrained trace variables.
For each trace variable in a set \(\Var\) and propositional variable, \(\traceVar\) and \(x\mathop{\in}\Prop\), \(\rewriteFlat\) introduces a new trace variable, \(\traceVar_x\), which is used exclusively in predicates involving the corresponding variables (e.g., \(x(\traceVar,i)\)).
For a given set of trace variables \(\Var\) and propositional variables \(\Prop\), we define \(\Var_{\Prop}\mathop{=}\{\traceVar_x\ |\ \traceVar \mathop{\in} \Var \tAnd x \mathop{\in} \Prop\}\).
\begin{align}
%---------- Trace Quantifiers
\rewriteFlat(\exists \traceVar\ \varphi, \{x_0, \ldots, x_n\},\VCons) = &
\exists \traceVar_{x_0} \ldots \exists \traceVar_{x_n} \rewriteFlat(\varphi,  \{x_0, \ldots, x_n\},\VCons\mathop{\cup}\{\traceVar\}) \nonumber\\
\rewriteFlat(\constrQ{\exists}{\traceVar}\ \varphi, \Prop,\VCons) = &
\constrQ{\exists}{\traceVar}\ \rewriteFlat(\varphi, \Prop,\VCons) \nonumber\\
%---------- Time Quantifiers
\rewriteFlat(\exists \timeVar\  \varphi, \Prop,\VCons) = &
\exists \timeVar\ \rewriteFlat(\varphi, \Prop,\VCons) \nonumber\\
%---------- Binary predicate
\rewriteFlat(x(\traceVar,i), \Prop, \VCons)= &
\begin{cases}
x(\traceVar_x,i) & \text{if } \traceVar\mathop{\in}\VCons\\
x(\traceVar,i) & \text{otherwise }
\end{cases}
\label{def:rewrite}
\\
%---------- Boolean operators
\rewriteFlat(\neg \varphi, \Prop,\VCons) = &
\neg \rewriteFlat(\varphi, \Prop,\VCons)
\nonumber\\
\rewriteFlat(\varphi \vee \varphi', \Prop,\VCons) = &
\rewriteFlat(\varphi, \Prop,\VCons)\vee \rewriteFlat(\varphi', \Prop,\VCons)
\nonumber
\end{align}
We prove below that \(\rewriteFlat\) returns an equisatisfiable hypertrace formula.

\begin{lemma}
\label{lemma:rewrite:independent}
Let \(\varphi\) be a hypertrace formula over the set of propositional variables \(\Prop\) and \(\setTraces\) be a set of traces over the same variables.
Let \(\AssignN\) and \(\AssignT\) be trace and time assignments 
over the set of free variables in \(\varphi\).
%, \(\Var\mathop{=}\freeV(\varphi)\).
%
Let \(\VCons\mathop{\subseteq}\freeV(\varphi)\) be a set of trace variables that are free in \(\varphi\).
For all trace assignments \(\AssignT'\) over trace variables in \(\VCons_{\Prop}\mathop{\cup}\freeV(\varphi)\) 
agreeing with  \(\AssignT\) in its assignments of propositional variables \(x\) for the trace assigned to \(\traceVar\) (i.e., for all \(\traceVar\mathop{\in}\VCons\) and \(x\mathop{\in}\Prop\),  \(\AssignT(\traceVar)\mathop{=}_{\{x\}} \AssignT'(\traceVar_x)\))
and otherwise having the same trace assignments of \(\AssignT\) (i.e., for all \(\traceVar\mathop{\notin}\VCons\), \(\AssignT(\traceVar)\mathop{=}\AssignT'(\traceVar)\)): % we have:
\((\overline{\setTraces},(\AssignT, \AssignN))\modelsTFOL\varphi \tIff (\overline{\setTraces},(\AssignT', \AssignN)) \modelsTFOL\) \(\rewriteFlat(\varphi, \Prop,\VCons).\)
\end{lemma}

\newcommand{\modelslhs}[2]{(\overline{\setTraces},(#1, #2))}
\newcommand{\xlhs}[1]{\modelslhs{\AssignT}{\AssignN}\modelsTFOL#1}
\newcommand{\xrhs}[1] {\modelslhs{\AssignT'}{\AssignN}\modelsTFOL\rewriteFlat(#1, \Prop,\VCons)}

\begin{proof}

Let \(\Prop\mathop{=}\{x_0, \ldots, x_n\}\).
We prove the statement by structural induction on the formula.

{
% First we observe that the direction \(\Rightarrow\) is straightforward:  we use the trace assigned to \(\traceVar\) in \(\AssignT\) in the assignment of \(\traceVar_x\) in \(\AssignT'\); that is, \(\AssignT(\traceVar)\mathop{=}_{\{x\}} \AssignT'(\traceVar_x)\).
}
Let us prove the \(\Rightarrow\) direction. In the base case ($\varphi = x(\pi, i)$),
$\xlhs{x(\pi, i)}$ iff $(\AssignT(\pi), \AssignN(i)) \mathop{\in} X_x$.
By definition of $X_x$ (given by $\overline{\setTraces}$), it holds that $(\AssignT(\pi), \AssignN(i)) \mathop{\in} X_x \Leftrightarrow
\AssignT(\pi)[\AssignN(i)](x) = {\mathit{true}}$.
If \(\traceVar\mathop{\in}\VCons\), this is equivalent to 
$\AssignT'(\pi_x)[\AssignN(i)](x) = {\mathit{true}}$ (by definition of $\AssignT'$).
Otherwise, it is  equivalent to 
$\AssignT'(\pi)[\AssignN(i)](x) = {\mathit{true}}$ (also by definition of $\AssignT'$).
Therefore, $\xrhs{x(\pi, i)}$.
If $\varphi$ is $i = j$ or $i < j$, the implication holds as these formulas depend only on $\AssignN$ which is not affected by changes from \(\rewriteFlat\).

In the induction step, cases for $\varphi_1 \lor \varphi_2$ and $\neg \varphi$ follow from the definition of $\modelsTFOL$ and induction hypothesis.
Assume the formula $\exists i\ \varphi$.
Then, $\xlhs{\exists i\ \varphi} \Leftrightarrow \exists c\mathop{\in}\nat$ s.t.
$\modelslhs{\AssignT}{\AssignN[i\mapsto c]} \modelsTFOL \varphi$.
From IH, $\modelslhs{\AssignT'}{\AssignN[i\mapsto c]} \modelsTFOL \rewriteFlat(\varphi, \Prop, \VCons)$
which implies  $\xrhs{\exists i\ \varphi}$.
Analogously for $\constrQ{\exists}{\pi}$. %\xxx{Write it down?}

Finally, assume formula $\exists \pi\ \varphi$ and a set of trace variables \(\VCons\mathop{\subseteq}\freeV(\exists \traceVar\ \varphi)\).
This formula is satisfied iff there exists $t$ s.t. $\modelslhs{\AssignT[\pi\mapsto t]}{\AssignN} \modelsTFOL \varphi$.
We observe that \(\traceVar\mathop{\in}\freeV(\varphi)\), as we assume there are no double binding of variables.
Then, from IH, for \({\VCons}'\mathop{=}\VCons\mathop{\cup}\{\traceVar\}\),
it follows that
$\modelslhs{\widetilde{\AssignT}}{\AssignN} \modelsTFOL \rewriteFlat(\varphi, \Prop, {\VCons}')$
where $\widetilde{\AssignT}$ is such that 
for all \(\traceVar'\mathop{\in}\VCons\cup\{\traceVar\}\) and \(x\mathop{\in}\Prop\),  \(\AssignT[\pi\mapsto t](\traceVar')\mathop{=}_{\{x\}} \widetilde{\AssignT}(\traceVar_x)\) and otherwise the trace assignments are the same as \(\AssignT\).
Observe that $\modelslhs{\AssignT'[\pi_{x_0}\mapsto \widetilde{\AssignT}(\pi_{x_0}), ..., \pi_{x_n}\mapsto \widetilde{\AssignT}(\pi_{x_n})]}{\AssignN} \modelsTFOL \rewriteFlat(\varphi, \Prop, \Var\mathop{\cup}\{\traceVar\})$.
From the definition of $\modelsTFOL$:
$\modelslhs{\AssignT'}{\AssignN} \modelsTFOL$ $\exists \pi_{x_0}...\exists\pi_{x_n}\rewriteFlat(\varphi, \Prop, \VCons)$.

% Now we prove the direction \(\Leftarrow\).
% %
% The base case follows from the assumption that \(\AssignT(\traceVar)\mathop{=}_{\{x\}} \AssignT'(\traceVar_x)\).
% %
% From the inductive cases, the only case that does not follow directly from the induction hypothesis is the unconstrained quantifier, which we prove next.
% %
% Consider an arbitrary hypertrace formula \(\exists \traceVar\ \varphi\) and trace assignment over its free variables \(\Var\) s.t.\ \(\AssignT(\traceVar')\mathop{=}_{\{x\}} \AssignT'(\traceVar'_x)\), for all \(\traceVar'\mathop{\in}\freeV(\Var)\) and \(x\mathop{\in}\Prop\).
% %
% Assume that \(\AssignT'\modelsTFOL\rewriteFlat(\exists \traceVar\ \varphi, \Prop)\).
% %
% By definition of \(\rewriteFlat\) and models relation, there exists \(\trace_{x_0}\ldots \trace_{x_n} \mathop{\in}(\AllVals)^\omega\) s.t.\ 
% \(\AssignT'[\traceVar_{x_0}\mapsto\trace_{x_0}, \ldots \traceVar_{x_n}\mapsto\trace_{x_n}] \modelsTFOL \varphi\).
% %
% Consider the trace \(\trace\) s.t.\, for all \(i\mathop{\in}\nat\) and \(x\mathop{\in}\Prop\), \(x\mathop{\in}\trace[i]\) iff \(x\mathop{\in}\trace_x[i]\).
% %
% Then, for all \(x\mathop{\in}\Prop\), \((\AssignT'[\traceVar_{x_0}\mapsto\trace_{x_0}, \ldots \traceVar_{x_n}\mapsto\trace_{x_n}])(\traceVar_x) \mathop{=}_{x} (\AssignT[\traceVar\mapsto\trace])(\traceVar)\) and,
% by induction hypothesis, \(\AssignT[\traceVar\mapsto\trace]\models \varphi\).

The proof for the $\Leftarrow$ direction is analogous to the $\Rightarrow$ direction, because the relation between $\AssignT$ and $\AssignT'$ is equational.
% We discuss only the induction step in the case of $\rewriteFlat(\exists \pi\ \varphi, \Prop)$.
% $\xrhs{\exists \pi\ \varphi}$ by definition only if
% $\modelslhs{\AssignT'}{\AssignN}\modelsTFOL\exists \pi_{x_0}...\pi_{x_n}\rewriteFlat(\varphi, \Prop)$
% which holds only if there exist $t_{x_0},...,t_{x_n} \in (\Bool^\Prop)^\omega$ such that
% $\modelslhs{\AssignT[\pi_{x_0}\mapsto t_{x_0},...,\pi_{x_n}\mapsto t_{x_n}]}{\AssignN}\modelsTFOL$ $\rewriteFlat(\varphi, \Prop)$.
% From IH, we get that then $\modelslhs{\widetilde{\AssignT}}{\AssignN} \modelsTFOL \varphi$
% where $\widetilde{\AssignT}$ is such that
% for all \(\traceVar'\mathop{\in}\Var\) and \(x\mathop{\in}\Prop\),  \(\AssignT(\traceVar')\mathop{=}_{\{x\}} \widetilde{\AssignT}(\traceVar'_x)\).
% Let $t$ be a trace such that $t\mathop{=}_{\{x\}} \widetilde{\AssignT}(\pi_x)$ for all $x\mathop{\in}\Prop$ ($\star$). 
% Note that $t$ is uniquely determined by $\widetilde{\AssignT}$.
% Then $\modelslhs{\AssignT[\pi\mapsto t]}{\AssignN}\modelsTFOL \varphi$ and thus $\xlhs{\exists \pi\ \varphi}$.
\end{proof}

It follows from the above lemma that all hypertrace formulas are equisatisfiable to their rewrite with \(\rewriteFlat\).

\begin{corollary}
Let \(\varphi\) be a closed hypertrace formula over the set of propositional variables \(\Prop\):
\(\generatedSet{\varphi}\mathop{\neq}\emptyset \tIff \generatedSet{\rewriteFlat(\varphi, \Prop,\emptyset)}\mathop{\neq}\emptyset\).
\end{corollary}

%-------------------------------------
\subsection{Satisfiability of Hypertrace Formulas}
%\label{sec:sat}

We are interested in the decidability of the satisfiability problem of hypertrace logic.

\begin{tcolorbox}[colback = gray!3, colbacktitle=gray!90, title={\textbf{Satisfiability of Hypertrace Formulas}}]
Let \(\varphi\) be a hypertrace formula over trace variables \(\VarTrace\), time variables \(\VarTime\) and binary predicates defined from \(\Prop\) (i.e., from the set \(\{\BinaryP_a \,|\, a\mathop{\in}\Prop\}\)).

Is there a set of traces \(\setTraces\mathop{\subseteq}(\AllVals)^\omega\) that is a model of \(\varphi\); that is, \(\generatedSet{\varphi}\mathop{\neq}\emptyset\)?
%\end{center}
\end{tcolorbox}

We present in Table \ref{tab:results} a summary of decidability results proved in this work for the satisfiability problem of hypertrace logic.
We describe fragments by specifying patterns on its quantifiers:
 we denote \(\Quant\mathop{\in}\{\forall, \exists\}\), while  \(\Quant_{\timeS}\) denotes time quantifiers, and \(\Quant_{\traceS}\) and \(\constrQ{}{\Quant_{\traceS}}\) denotes unconstrained and constrained quantifiers, respectively.
We  then combine these symbols to define patterns as regular expressions.

{\centering
\setlength{\tabcolsep}{10pt}
\begin{table}[h!]
\centering
\begin{tabular}{ll}
\hline
\rowcolor{gray!20}
\multicolumn{2}{c}{\textbf{Trace-prefixed}} \\
\hline
\(\exists_{\traceS}^*
(\constrQ{\exists_{\traceS}}{})^*
(\constrQ{\forall_{\traceS}}{})^*
\Quant_{\traceS}^* 
\Quant_{\timeS}^*\) & 
Decidable \cite{Hahn21} [Cor. \ref{thm:trace_prefic_dec}]
\\
%%%%%%%%%%%%%%%%%%%%%%%%%%%%%%%%%%%
\((\constrQ{\forall_{\traceS}}{})^2 \constrQ{\exists_{\traceS}}{}\ \Quant_{\timeS}^+ \) & 
Undecidable  \cite{satHyperLTL16} [Prop. \ref{thm:undec_traceprefix}]
\\
\hline
\rowcolor{gray!20}
\multicolumn{2}{c}{\textbf{Time-prefixed}} \\
\hline
%%%%
\(\exists_{\timeS}^*
\QExists_{\traceS}^*
(\constrQ{\exists_{\traceS}}{})^*
(\constrQ{\forall_{\traceS}}{})^*\Quant_{\traceS}^*\) & Decidable [Cor. \ref{thm:time_pre_dec}]\\

\(\exists_{\timeS}
\forall_{\timeS}
\exists_{\timeS}^2
\forall_{\timeS}
\constrQ{\forall_{\traceS}}{}
(\constrQ{\exists_{\traceS}}{})^2
\exists_{\traceS}\) & 
Undecidable [Thm.~\ref{thm:time_prefix_undec}]\\
%%%%%%%%%%%%%%%%%%%%%%%%%%%%%
\bottomrule
\end{tabular}
\caption{Summary of results on the decidability of hypertrace formulas satisfiability.}
\label{tab:results}
\end{table}
\vspace{-5mm}
}

\vspace{-2mm}
\section{Unconstrained Hypertrace Logic}
\label{sec:unconstr}
\vspace{-2mm}

We look into the fragment of hypertrace formulas without constrained quantifiers (i.e., no subformulas with shape \(\constrQ{\Quant}{\traceVar}\ \varphi\)), which we refer to as \emph{unconstrained hypertrace logic}.
We prove that unconstrained hypertrace logic is expressively equivalent to monadic second-order logic of order of one successor (\SOneS).
 The translation follows naturally from Lemma \ref{lemma:rewrite:independent}.
 %, where we prove that we can rewrite hypertrace formulas to having uncontrained quant
% Intuitively, assignments to variables in unconstrained traces are entirely independent of one another.
% \xxx{I do not share this intuition, the quantifiers are still dependent. If you say ,,for any trace t, there exists trace t' where first symbol is 'a', then t' depends on t.}
%
Specifically, we show that quantification over variables in each unconstrained trace is equivalent to second-order quantification over the sets of time points at which the corresponding propositional variable holds in that trace.

\vspace{-2mm}
\subsection{Equivalence to S1S}
\label{sec:uncons_s1s}
\vspace{-2mm}

%We observe that the rewrite defined by \(\rewriteFlat\) does not produce equisatisfiable formulas for constrained hypertrace formulas because the trace formed by joining the valuations of different propositional variables may not define a trace in the set of traces given as a model.
\medskip
To establish the equivalence between unconstrained hypertrace logic and \(\SOneS\), we define a translation between unconstrained hypertrace formulas and models to their counterpart in \(\SOneS\), and vice-versa.
The translation from unconstrained hypertrace formulas to \(\SOneS\) rewrites the formulas using \(\rewriteFlat\), while reinterpreting each \(\traceVar_x\) as a second-order variable.
% {\color{red}
% %
% \begin{equation}
% \label{eq:toMSO}
% \toMSO(\varphi, \Prop)=
%     \begin{cases}
%     \traceVar_x(i) & \tIf \varphi\mathop{=}x(\traceVar,i)\\
%     \rewriteFlat(\varphi, \Prop) & \text{ otherwise}.
%     \end{cases}
% \end{equation}
% \xxx{This definition does not work as intended, so I rewrote it to use substitution instead}
% }

To translate from unconstrained formulas to \(\SOneS\) formulas, we 
define the substitution 
$\sigma = \{x(\traceVar_x, i) \mapsto \traceVar_x(i) \mid x\in\Prop,\pi\in\Var_{\traceS} , i\in\Var_{\nat}\} $
and the rewriting function
%
%\begin{equation}
%\label{eq:toMSO}
\(\toMSO(\varphi, \Prop)=
    \rewriteFlat(\varphi, \Prop,\Var_{\traceS}\mathop{\cap}\freeV(\varphi))[\sigma]\) where \(\varphi\)
    is an unconstrained trace formula.
%\end{equation}
%
{This translation does not work for formulas with constrained trace quantifiers because}, in general, we cannot assume $T$ to be $\SOneS$-definable.
The next step is to define a translation for trace assignments.
A {trace} assignment \(\AssignT\As
\Var_{\traceS} \rightarrow (\AllVals)^{\omega}\) is translated to an assignment over second-order variables where each variable is mapped to its support set. Formally,
\(\toTime(\AssignT)\As
\Var_{\Prop} \rightarrow \mathbf{2}^{\nat}\) where
\(
\toTime(\AssignT)(\traceVar_x) = \{i \ |\ x \mathop{\in} \AssignT(\traceVar)[i]\del{ = 1}\}.
\)

For the translation from \(\SOneS\) to unconstrained hypertrace formulas, each second-order variable \(X\) becomes the trace variable \(\trace_X\). Additionally, we use the standard translation of \(\Succ\) using \(\leq\).
%
%Formally, for a \(\SOneS\) formula \(\varphi\) we define three substitutions:
%\[\sigma_X = \{X \mapsto \traceVar_X\ \mid X \in \Var_2\}\hfill X\]
\begin{align}
\label{eq:toHyper}
\toHyper(\exists X\ \varphi')=& \exists \traceVar_X\ \toHyper(\varphi') & \toHyper(X(i))=& X(\traceVar_X,i)\nonumber\\
\toHyper(\exists i\ \varphi')=& \exists i\ \toHyper(\varphi')& \toHyper(i = j)=& (i = j)\nonumber\\
\toHyper(\neg \varphi')=& \neg \toHyper(\varphi')\\
\toHyper(\varphi_1 \lor \varphi_2)=& \toHyper(\varphi_1) \lor \toHyper(\varphi_2)\nonumber\\
\toHyper(\Succ(i,i'))=&  {i<i'} \wedge (\forall j\ (i < j \rightarrow i' \leq j))\nonumber
\end{align}
We translate second-order assignments, \(\AssignSecond\), to trace assignments as:
\(\toBool(\AssignSecond)\As
\Var_{\traceS} \rightarrow (\AllVals)^{\omega}\) where
for all \(i\mathop{\in}\nat\) and \(X\mathop{\in}\Prop\),
\(X\mathop{\in}(\toBool(\AssignSecond)(\traceVar_X))[i]\) iff \(i\mathop{\in}\AssignSecond(X).
\)

\begin{theorem}
\label{thm:equiv:unconstrained-s1s}
Let \(\setTraces\) be a set of traces over \(\Prop\).
For all unconstrained hypertrace formulas \(\varphi\) over \(\Prop\):
\((\overline{\setTraces},(\AssignT, \AssignN)) \modelsTFOL \varphi\) iff 
\((\AssignN, \toTime(\AssignT)) \modelsSOneS \toMSO(\varphi, \Prop)\).
For all \(\SOneS\) formulas: 
\((\AssignFirst, \AssignSecond) \modelsSOneS \varphi\) iff
\((\overline{\setTraces},(\toBool(\AssignSecond), \AssignFirst)) \modelsTFOL \toHyper(\varphi)\).
\end{theorem}

\begin{proof}
We start by proving, by structural induction, the translation from unconstrained hypertrace formulas and trace assignments to \(\SOneS\) formulas and second-order assignments.
For the base case, we need to prove \((\overline{\setTraces},(\AssignT, \AssignN)) \modelsTFOL x(\traceVar,i)\) iff 
\((\AssignN, \toTime(\AssignT)) \modelsSOneS \traceVar_x(i)\).
Given \(k\mathop{=}\AssignN(i)\),
we observe that, by definition,  \(x\mathop{\in} \AssignT(\traceVar)[k]\) iff \(k\mathop{\in}(\toTime(\AssignT))(\traceVar_x)\).
Hence, the base case holds.
For the induction cases, we use the induction hypothesis and
Lemma~\ref{lemma:rewrite:independent}.

We consider now the translation from \(\SOneS\) formulas and second-order assignments to unconstrained hypertrace formulas and trace assignments.
We have two base cases. We start by proving that 
\((\AssignFirst, \AssignSecond) \modelsSOneS X(i)\) iff
\((\overline{\setTraces},(\toBool(\AssignSecond),\AssignFirst) \modelsTFOL X(\traceVar_X,i)\).
This holds, because for \(k\mathop{=}\AssignN(i)\), 
\(X\mathop{\in}(\toBool(\AssignSecond)(\traceVar_X))[k]\) iff \(k\mathop{\in}\AssignSecond(X).\)
The second base-case is \((\AssignFirst, \AssignSecond) \modelsSOneS \Succ(i\diff{.}{,}i')\) iff
\((\overline{\setTraces},(\toBool(\AssignSecond),\AssignFirst)) \modelsTFOL {i<i'} \wedge (\forall j\ (i < j \rightarrow i' \leq j))\), holding by definition of \(\leq\).
For the induction cases, we use the induction hypothesis and
Lemma~\ref{lemma:rewrite:independent}.
\end{proof}

% \ana{To define language equivalence we move this here:

% Concretely,
% each trace assignment \(\AssignT\As (\Var_{\traceS} \rightarrow (\AllVals)^{\omega})\) defines the trace \(\toTraceH{\AssignT}\) over \(\Prop_{\VarTrace} \mathop{=} \{ x_\traceVar \,|\, x \in \Prop \tAnd \traceVar \in \VarTrace\}\) where
% %\begin{align*}
% \(
% \text{for all } k\mathop{\in}\nat \tAnd x_{\traceVar}\mathop{\in}\Prop_{\VarTrace}:
%     x_{\traceVar}\mathop{\in}\toTraceH{\AssignT}[k] \tIff  
%     x\mathop{\in}\AssignT[k].
% \)

% A pair of assignments \((\AssignFirst, \AssignSecond)\) defines the trace \(\toTraceS{\AssignFirst, \AssignSecond}\) over 
% \(\Bool^{\Var_1 \mathop{\cup} \Var_2}\) where, 
% %
% for all first-order variables \(i\mathop{\in}\Var_1\) and 
% second-order variables \(\UnaryP\mathop{\in}\Var_2\):
% \begin{align*} 
% i\mathop{\in}\toTraceS{\AssignFirst, \AssignSecond}[k] \tIff \AssignFirst(i)\mathop{=}k 
% %\\ 
% \hspace{5mm}
% \tAnd
% \hspace{5mm}
% \UnaryP\mathop{\in}\toTraceS{\AssignFirst, \AssignSecond}[k] \tIff k\mathop{\in}\AssignSecond(\UnaryP).
% \end{align*}

% }

From the previous theorem and \(\SOneS\) having a decidable satisfiability checking problem~\cite{S1SBuchi60}, checking whether an unconstrained hypertrace formula is satisfiable is also decidable.

\begin{corollary}
\label{cor:unconst:dec}
The satisfiability problem for unconstrained hypertrace logic is decidable.
\end{corollary}

\vspace{-2mm}
\subsection{Relation to Full Hypertrace Logic}
\label{sec:unconst_vs_full}
\vspace{-2mm}
\newcommand{\tfree}{\text{qf}}
\newcommand{\flatAll}{\mathtt{removeForAll}}

In this section, we show that any hypertrace formula with a single alternation from an existential to a universal constrained quantifier can be rewritten into an equisatisfiable formula without constrained quantifiers.
We prove our results by adapting techniques introduced in \cite{satHyperLTL16,hierarchyHyper19} to rewrite \(\HQPTL\) formulas into \(\QPTL\).
We start by showing how to eliminate the universal constrained trace quantifiers.

\begin{lemma}
\label{lemma:removeAll}
Let \(\varphi = \overrightarrow{\QExists} \
\constrQ{\exists}{\traceVar_1} \LLL \constrQ{\exists}{\traceVar_n}\
\constrQ{\forall}{\traceVar'_1}\LLL \constrQ{\forall}{\traceVar'_m}\
\overrightarrow{\Quant} \
\varphi_{\tfree}\)
be a hypertrace formula in prenex normal form s.t.\ 
\(\QExists\) is any combination of existential time or unconstrained trace quantifiers,
{\(\overrightarrow{\Quant}\)}
is any combination of time or unconstrained trace quantifiers and \(\varphi_{\tfree}\) is a quantifier-free hypertrace formula.
Let
\begin{align*}
\hspace{-5mm}
\flatAll(\varphi) = \overrightarrow{\QExists} \
\constrQ{\exists}{\traceVar_1} \LLL \constrQ{\exists}{\traceVar_n}\
\bigwedge\limits_{j_1=1}^n \LLL \bigwedge\limits_{j_m=1}^n
\overrightarrow{\Quant} \
\varphi_{\tfree}[\traceVar'_1\mapsto\traceVar_{j_1}, \LLL, \traceVar'_m\mapsto\traceVar_{j_m}]
\end{align*}
where \(\varphi_{\tfree}[\traceVar'_1\mapsto\traceVar_{j_1}, \LLL, \traceVar'_m\mapsto\traceVar_{j_m}]\) is the formula obtained by substituting \(\traceVar'_i\) with \(\traceVar_{j_i}\), for all \(1\mathop{\leq}i\mathop{\leq}m\), in \(\varphi_{\tfree}\).
%\(\setTraces\models\varphi\) iff there exists \(\setTraces'\mathop{\subseteq}\setTraces\) s.t.\ 
%\(\setTraces'\models\varphi\) iff \(\setTraces'\models \flatAll(\varphi)\).
Then,
\(\generatedSet{\varphi}\neq \emptyset\) iff \(\generatedSet{\flatAll(\varphi)}\neq\emptyset\).
\end{lemma}

\begin{proof}
%
% Let \(\varphi \mathop{=} \overrightarrow{\QExists} \
% \constrQ{\exists}{\traceVar_1} \ldots \constrQ{\exists}{\traceVar_n}\
% \constrQ{\forall}{\traceVar'_1}\ldots \constrQ{\forall}{\traceVar'_m}\
% \overrightarrow{\Quant} \
% \varphi_{\tfree}\) where 
% \(\QExists\) is any combination of existential time or unconstrained trace quantifiers,
% {\(\overrightarrow{\Quant}\)}
% are any combination of time or unconstrained trace quantifiers and \(\varphi_{\tfree}\) is a quantifier-free hypertrace formula.
%
A set of trace variables \(\Var\) and a trace assignment {\(\AssignT\)} that includes assignments for all variables in \(\Var\) 
(i.e., \(\Var\mathop{\subseteq}\text{Dom}({\AssignT})\)) 
define the set of traces 
\(\setTraces_{(\AssignT,\Var)}\mathop{=}\{\AssignT(\traceVar) \mid \traceVar\mathop{\in}\Var\}\).
For a set of free variables and assignment, \(\Var\) and \(\AssignT\), when we evaluate the subformula starting with the universal constrained quantifier against the set of traces \(\setTraces_{(\AssignT,\Var)}\),
we can remove the universal quantifiers by considering all possible combinations of assignments to the traces in \(\setTraces_{(\AssignT,\Var)}\).
Formally, for any hypertrace formula 
\(\constrQ{\forall}{\traceVar'_1}\ldots \constrQ{\forall}{\traceVar'_m}\
\overrightarrow{\Quant} \
\varphi_{\tfree}\) where 
 \(\overrightarrow{\Quant}\)
is any combination of time or unconstrained trace quantifiers, \(\varphi_{\tfree}\) is a quantifier-free hypertrace formula, and 
\(\Var\mathop{=}\{\traceVar_1, \ldots, \traceVar_n\}\) s.t.\ 
\(\Var\mathop{\subseteq} \freeV(\constrQ{\forall}{\traceVar'_1}\ldots \constrQ{\forall}{\traceVar'_m}\
\overrightarrow{\Quant} \
\varphi_{\tfree})\).
Formally:
\begin{align}
&(\overline{\setTraces_{(\AssignT,\Var)}}, (\AssignT,\AssignN))\modelsTFOL \constrQ{\forall}{\traceVar'_1}\ldots \constrQ{\forall}{\traceVar'_m}\
\overrightarrow{\Quant} \
\varphi_{\tfree} \tIff \label{lemma:removeAll:1}\\
& (\overline{\setTraces_{(\AssignT,\Var)}}, (\AssignT,\AssignN))\modelsTFOL
\bigwedge\limits_{j_1=1}^n \ldots \bigwedge\limits_{j_m=1}^n
\overrightarrow{\Quant} \
\varphi_{\tfree}[\traceVar'_1\mapsto\traceVar_{j_1}, \LLL, \traceVar'_m\mapsto\traceVar_{j_m}].\nonumber
\end{align}
We prove this equivalence below:
 \begin{align*}
 &(\overline{\setTraces_{(\AssignT,\Var)}}, (\AssignT,\AssignN))\modelsTFOL \constrQ{\forall}{\traceVar'_1}\ldots \constrQ{\forall}{\traceVar'_m}\
 \overrightarrow{\Quant} \
 \varphi_{\tfree} \overset{\text{def. } \modelsTFOL}{\Leftrightarrow}\\
 & \text{for all } \trace_1\mathop{\in}\setTraces_{(\AssignT,\Var)}, \ldots, 
 \trace_m\mathop{\in}\setTraces_{(\AssignT,\Var)}:\\
 &\hspace{10mm}
 (\overline{\setTraces_{(\AssignT,\Var)}}, (\AssignT[\traceVar_1 \mapsto\trace_1,\LLL,\traceVar_m \mapsto\trace_m],\AssignN))\modelsTFOL 
 \overrightarrow{\Quant} \
 \varphi_{\tfree} \overset{\text{def. } \setTraces_{(\AssignT,\Var)}}{\Leftrightarrow}\\
 & \text{for all } \trace_1\mathop{\in}\{\AssignT[\traceVar_1],\LLL,\AssignT[\traceVar_n]\}, \ldots, 
 \trace_m\mathop{\in}\{\AssignT[\traceVar_1],\LLL,\AssignT[\traceVar_n]\}:\\
 &\hspace{10mm}
 (\overline{\setTraces_{(\AssignT,\Var)}}, (\AssignT[\traceVar_1 \mapsto\trace_1,\LLL,\traceVar_m \mapsto\trace_m],\AssignN))\modelsTFOL 
 \overrightarrow{\Quant} \
 \varphi_{\tfree} 
 \overset{}{\Leftrightarrow}\\
 & \text{for all } \traceVar_{j_1} \mathop{\in}\{1,\LLL,n\}, \LLL,\traceVar_{j_m} \mathop{\in}\{1,\LLL,n\}:\\
 &\hspace{10mm}
 (\overline{\setTraces_{(\AssignT,\Var)}}, (\AssignT[\traceVar_1 \mapsto\AssignT[\traceVar_{j_1}],\LLL,\traceVar_m \mapsto\AssignT[\traceVar_{j_m}]],\AssignN))\modelsTFOL 
 \overrightarrow{\Quant} \
 \varphi_{\tfree} 
 \overset{\text{finite domains}}{\Leftrightarrow}\\
 &\bigwedge\limits_{j_1=1}^n \ldots \bigwedge\limits_{j_m=1}^n: 
 (\overline{\setTraces_{(\AssignT,\Var)}}, (\AssignT[\traceVar_1 \mapsto\AssignT[\traceVar_{j_1}],\LLL,\traceVar_m \mapsto\AssignT[\traceVar_{j_m}]],\AssignN))\modelsTFOL 
 \overrightarrow{\Quant} \
 \varphi_{\tfree} 
 \overset{\text{def. }\modelsTFOL}{\Leftrightarrow}\\
 & (\overline{\setTraces_{(\AssignT,\Var)}}, (\AssignT,\AssignN))\modelsTFOL
 \bigwedge\limits_{j_1=1}^n \ldots \bigwedge\limits_{j_m=1}^n
 \overrightarrow{\Quant} \
 \varphi_{\tfree}[\traceVar'_1\mapsto\traceVar_{j_1}, \LLL, \traceVar'_m\mapsto\traceVar_{j_m}].
 \end{align*}

Before we prove the main claim, we need to prove that the set of models of a hypertrace formulas without existential constrained quantifiers is subset-closed.

{\bf Claim 1}: \emph{Let \(\varphi\) be a hypertrace formula in prenex and negation normal form (i.e., all quantifiers at the beginning of the formula and negation only at the atomic level) without existential constrained trace quantifiers.
For all sets of traces \(\setTraces\) and its subsets \(\setTraces'\mathop{\subseteq}\setTraces\), and assignments \(\AssignT\) and 
\(\AssignN\): if \((\overline{\setTraces}, (\AssignT,\AssignN) \modelsTFOL \varphi\), then 
\((\overline{\setTraces'}, (\AssignT,\AssignN) \modelsTFOL \varphi\).}

{\bf Proof of claim 1}: We prove by structural induction on \(\varphi\).
For the bases cases, we observe that \(\modelsTFOL\) is independent of the set of traces, hence, as the assignments are preserved, the property holds.
The induction cases \(\varphi\vee\varphi'\), \(\varphi\wedge\varphi'\), \(\exists i\ \varphi\) and \(\forall i \ \varphi\) follow directly from definitions and induction hypothesis.
The only case remaining is \(\constrQ{\forall}{\traceVar}\ \varphi\).
Consider arbitrary set of traces \(\setTraces\) and \(\setTraces'\mathop{\subseteq}\setTraces\) and assignments \(\AssignT\) and 
\(\AssignN\).
We assume that  \((\overline{\setTraces}, (\AssignT,\AssignN) \modelsTFOL \constrQ{\forall}{\traceVar}\ \varphi\), which, by definition of \(\modelsTFOL\), is equivalent to:
for all traces \(\trace\mathop{\in}\setTraces\), \((\overline{\setTraces}, (\AssignT[\traceVar\mapsto\trace],\AssignN) \modelsTFOL \varphi\).
By induction hypothesis and \(\setTraces'\mathop{\subseteq}\setTraces\), 
for all traces \(\trace\mathop{\in}\setTraces'\), \((\overline{\setTraces'}, (\AssignT[\traceVar\mapsto\trace],\AssignN) \modelsTFOL \varphi\).
Hence, by definition of \(\modelsTFOL\), \((\overline{\setTraces'}, (\AssignT,\AssignN) \modelsTFOL \constrQ{\forall}{\traceVar}\ \varphi\).
\medskip

\newcommand{\VExist}{\Var_{\exists}}
We prove now each side of the implication.
We assume without loss of generality that the formula is in prenex and negation normal form.
In what follows, \(\VExist\mathop{=}\{\traceVar_1, \LLL, \traceVar_n\}\).
We start by proving that \(\generatedSet{\varphi}\mathop{\neq} \emptyset \Rightarrow \generatedSet{\flatAll(\varphi)}\mathop{\neq} \emptyset\):
\begin{align*}
	&\generatedSet{\varphi}\mathop{\neq} \emptyset \Leftrightarrow 
	\text{exists } \setTraces \mathop{\in} \generatedSet{\overrightarrow{\QExists} \
		\constrQ{\exists}{\traceVar_1} \ldots \constrQ{\exists}{\traceVar_n}\
		\constrQ{\forall}{\traceVar'_1}\ldots \constrQ{\forall}{\traceVar'_m}\
		\overrightarrow{\Quant} \
		\varphi_{\tfree}} 
	\overset{\text{def. } \modelsTFOL}{\Leftrightarrow}\\
	& \text{exists } \setTraces, \AssignT \tAnd \AssignN \tSt \AssignT(\traceVar)\mathop{\in}\setTraces \text{ with } \traceVar\mathop{\in}\VExist:\\
	&
	\hspace{10mm}
	(\overline{\setTraces}, (\AssignT, \AssignN))\modelsTFOL \ \constrQ{\forall}{\traceVar'_1}\ldots \constrQ{\forall}{\traceVar'_m}\ \overrightarrow{\Quant}\ \varphi_{\tfree} 
	\overset{\setTraces_{(\AssignT,\VExist)}\mathop{\subseteq}\setTraces, \text{Claim 1}}{\Rightarrow}\\
	&\text{exists } \AssignT \tAnd \AssignN:
	(\overline{\setTraces_{(\AssignT,\VExist)}}, (\AssignT, \AssignN))\modelsTFOL  \constrQ{\forall}{\traceVar'_1}\ldots \constrQ{\forall}{\traceVar'_m}\ \overrightarrow{\Quant}\ \varphi_{\tfree} 
	\overset{(\ref{lemma:removeAll:1})}{\Leftrightarrow}\\
	&\text{exists } \AssignT \tAnd \AssignN:\\
	&
	\hspace{10mm}
	(\overline{\setTraces_{(\AssignT,\VExist)}}, (\AssignT, \AssignN))\modelsTFOL \bigwedge\limits_{j_1=1}^n \ldots \bigwedge\limits_{j_m=1}^n
	\overrightarrow{\Quant} \
	\varphi_{\tfree}[\traceVar'_1\mapsto\traceVar_{j_1}, \LLL, \traceVar'_m\mapsto\traceVar_{j_m}]
		\overset{}{\Leftrightarrow}\\
	&\text{exists } \AssignT \tAnd \AssignN:
	\setTraces_{(\AssignT,\VExist)}\modelsTFOL \flatAll(\varphi)
	\quad \Rightarrow \quad \generatedSet{\flatAll(\varphi)}\mathop{\neq} \emptyset
\end{align*}

And now, we prove \(\generatedSet{\flatAll(\varphi)}\mathop{\neq} \emptyset \Rightarrow \generatedSet{\varphi}\mathop{\neq} \emptyset\):
\begin{align*}
	&\generatedSet{\flatAll(\varphi)}\mathop{\neq} \emptyset \Leftrightarrow\\
	& \text{exists } \setTraces \mathop{\in} \generatedSet{\overrightarrow{\QExists} \
		\constrQ{\exists}{\traceVar_1} \LLL \constrQ{\exists}{\traceVar_n}\!\!
		\bigwedge\limits_{j_1=1}^n \LLL \bigwedge\limits_{j_m=1}^n\!\!
		\overrightarrow{\Quant} \
		\varphi_{\tfree}[\traceVar'_1\mapsto\traceVar_{j_1}, \LLL, \traceVar'_m\mapsto\traceVar_{j_m}]} 
	 \overset{\text{def. } \modelsTFOL}{\Leftrightarrow}\\
	 &\text{exists } \setTraces, \AssignT \tAnd \AssignN:
	 (\overline{\setTraces}, (\AssignT, \AssignN))\modelsTFOL\!\! \bigwedge\limits_{j_1=1}^n \LLL \bigwedge\limits_{j_m=1}^n\!\!
	 \overrightarrow{\Quant} \
	 \varphi_{\tfree}[\traceVar'_1\mapsto\traceVar_{j_1}, \LLL, \traceVar'_m\mapsto\traceVar_{j_m}] 
	\underset{\text{Claim 1, (8)}}{\overset{\setTraces_{(\AssignT,\VExist)}\mathop{\subseteq}\setTraces}{\Rightarrow}}\\
	 &\text{exists }  \AssignT \tAnd \AssignN:
	 (\overline{\setTraces_{(\AssignT,\VExist)}}, (\AssignT, \AssignN))\modelsTFOL \!\!\! \bigwedge\limits_{j_1=1}^n \LLL \bigwedge\limits_{j_m=1}^n\!\!
	 \overrightarrow{\Quant} \
	 \varphi_{\tfree}[\traceVar'_1\mapsto\traceVar_{j_1}, \LLL, \traceVar'_m\mapsto\traceVar_{j_m}] 
	 \Leftrightarrow\\
	&\text{exists } \AssignT \tAnd \AssignN:
	(\overline{\setTraces_{(\AssignT,\VExist)}}, (\AssignT, \AssignN))\modelsTFOL  \constrQ{\forall}{\traceVar'_1}\ldots \constrQ{\forall}{\traceVar'_m}\ \overrightarrow{\Quant}\ \varphi_{\tfree} 
	\overset{}{\Rightarrow}\text{exists } \setTraces \mathop{\in}\generatedSet{\varphi} \qedhere
\end{align*}
\end{proof}

\begin{example}
for the formula:
$\varphi = \exists\pi_1\exists\pi_2\forall \pi_3\exists i\exists j\ a(\pi_1, i) \land \neg a(\pi_2, i) \land b(\pi_3, j)$,
we have 
\begin{flalign*}
\flatAll&(\varphi) = \exists\pi_1\exists\pi_2
\\&
(\exists i\exists j\ a(\pi_1, i) \land \neg a(\pi_2, i) \land b(\pi_1, j)) \land 
%\\&
(\exists i\exists j\ a(\pi_1, i) \land \neg a(\pi_2, i) \land b(\pi_2, j))
\end{flalign*}
which can be simplified to
$\exists\pi_1\exists\pi_2\exists i\ a(\pi_1, i) \land \neg a(\pi_2, i) \land (\exists j\ b(\pi_1, j))
\land (\exists j\ b(\pi_2, j))$.
\end{example}

We observe that hypertrace formulas with only existential constrained trace quantifiers are equisatisfiable to an unconstrained formula.
Intuitively, each existential trace quantifier can be instantiated independently of the others.

\begin{lemma}
\label{lemma:removeExists}
Let \(\varphi = \overrightarrow{\QExists} \
\constrQ{\exists}{\traceVar_1} \ldots \constrQ{\exists}{\traceVar_n}\
\overrightarrow{\Quant} \
\varphi_{\tfree}\)
be a hypertrace formula in prenex normal form s.t.\ 
\(\overrightarrow{\QExists}\) is a any combination of existential time or unconstrained trace quantifiers,
\(\overrightarrow{\Quant}\) is any combination of time or unconstrained trace quantifiers and \(\varphi_{\tfree}\) is a quantifier-free hypertrace formula.
\(\generatedSet{\varphi}\neq\emptyset\) iff 
\(\generatedSet{\overrightarrow{\QExists} \
\exists\traceVar_1 \ldots \exists \traceVar_n\
\overrightarrow{\Quant} \
\varphi_{\tfree}}\neq \emptyset\).
\end{lemma}

Finally, for all hypertrace formulas where all universal constrained trace quantifiers are only followed by existential constrained quantifiers, we can apply the rewrite described in Lemma 
\ref{lemma:removeAll} and Lemma \ref{lemma:removeExists} to get an equisatisfiable unconstrained hypertrace formula.

\begin{theorem}
\label{thm:toUnconstrained}
Let \(\varphi \mathop{=} \overrightarrow{\QExists} \
\constrQ{\exists}{\traceVar_1} \LLL \constrQ{\exists}{\traceVar_n}\
\constrQ{\forall}{\traceVar'_1}\LLL \constrQ{\forall}{\traceVar'_m}\
\overrightarrow{\Quant} \
\varphi_{\tfree}\)
be a hypertrace formula in prenex normal form s.t.\ 
\(\QExists\) is any combination of existential time or unconstrained trace quantifiers,
{\(\overrightarrow{\Quant}\)}
is any combination of time or unconstrained trace quantifiers and \(\varphi_{\tfree}\) is a quantifier-free hypertrace formula.
There exists an unconstrained hypertrace formula \(\varphi_{u}\) s.t.:
\(\generatedSet{\varphi}\neq \emptyset\) iff \(\generatedSet{\varphi_{u}}\neq \emptyset\).
\end{theorem}

\section{Trace-prefixed Hypertrace Logic}
\label{sec:trace_prefix}

In this section, we consider the fragment of \(\logic\) where trace quantifiers come before time quantifiers.
We call this the \emph{trace-prefixed hypertrace logic}, denoted
\(\TraceTFOL\).
Formally, trace-prefixed  formulas \(\varphi \in \TraceTFOL\) are defined by the grammar:
\[
\begin{split}
    \varphi ::= \exists \traceVar\, \varphi \ |\  \constrQ{\exists}{\traceVar}\, \varphi \ |\ \neg \varphi \ |\  \psi 
    \hspace{10mm}
	\psi &::= \exists \timeVar\ \psi \ |\ \psi \vee \psi \ |\ \neg \psi  \ |\  \timeVar < \timeVar\ |\  \timeVar = \timeVar \ |\    \BinaryP(\traceVar,\timeVar)
\end{split}
\] 
where \(\traceVar\) is a trace variable,
\(\timeVar\) is a time variable and \(\BinaryP\) is a binary predicate.

We observe that \(\TraceTFOL\)
is orthogonal to unconstrained hypertrace logic. 
While the unconstrained fragment allows any mix of (unconstrained) trace and time quantifiers, \(\TraceTFOL\) requires trace quantifiers to come first.
On the other hand, \(\TraceTFOL\) introduces constrained trace quantifiers, which the unconstrained version does not support.

\begin{example}
We can express \emph{bounded promptness} in \(\TraceTFOL\):
\begin{align}
\label{ex:bounded_promp}
\exists \traceVar\ \constrQ{\forall}{\traceVar'}\ \exists i\ \forall j\ \big( (j\mathop{<}i \rightarrow \neg p(\traceVar,j)) \wedge p(\traceVar,i) \wedge q(\traceVar',i)\big).
\end{align}
In this property, the unconstrained trace is used to guess a synchronization point (i.e., the time when \(p\) becomes true in \(\traceVar\)) where all traces agree with the value of proposition \(q\).
\end{example}

In \cite{bartocci2022flavors}, the authors study the trace-prefixed hypertrace logic for the fragment of \(\logic\) with only constrained trace quantifiers, which we refer to as \(\cTraceTFOL\).
They prove that \(\cTraceTFOL\) is expressively equivalent to \(\HLTL\).
From this result, we can prove that adding unconstrained quantifiers to \(\cTraceTFOL\) extends its expressive power.

\begin{proposition}
\label{thm:const_vs_full:trace}
\(\TraceTFOL\) is strictly more expressive than \(\cTraceTFOL\).
\end{proposition}
\begin{proof}
Clearly, \(\cTraceTFOL\) is a fragment of \(\TraceTFOL\).
To prove that  \(\cTraceTFOL\) is not equally expressive to \(\TraceTFOL\), we observe that bounded promptness (\ref{ex:bounded_promp}) is not expressible in \(\HLTL\) \cite{bozzelli2015unifying} and, from \(\cTraceTFOL\) being equally expressive to \(\HLTL\) \cite{bartocci2022flavors}, it is also not expressible in \(\cTraceTFOL\).
\end{proof}

We can also use \(\cTraceTFOL\) to get our first undecidability result for \(\TraceTFOL\).

\begin{proposition}
\label{thm:undec_traceprefix}
Let \(\varphi \mathop{=} 
\constrQ{\forall}{\traceVar_1}\ \constrQ{\forall}{\traceVar_2}\
\constrQ{\exists}{\traceVar'_1}\
\psi\) where \(\psi\) has only time quantifiers.
It is undecidable to check whether \(\generatedSet{\varphi}\mathop{\neq}\emptyset\).
\end{proposition}
\begin{proof}
We prove by reduction of the \(\HLTL\) satisfiability problem for formulas with quantifier pattern \(\forall\forall\exists\) to the problem of checking for satisfiability of hypertrace formulas with quantifier pattern \(\constrQ{\forall}{}\ \constrQ{\forall}{}\ \constrQ{\exists}{}\).
Consider an arbitrary \(\HLTL\) formula \(\varphi\mathop{=} \forall \traceVar_1 \forall \traceVar_2 \exists \traceVar_1' \psi\) where \(\psi\) has only temporal operators.
By \(\cTraceTFOL\) being expressively equivalent to \(\HLTL\), there exists \(\varphi'\mathop{\in}\cTraceTFOL\) s.t.\
\(\generatedSet{\varphi}\mathop{=}\generatedSet{\varphi'}\).
Using the computable translation in \cite{bartocci2022flavors}, 
\(\varphi'\mathop{=} \constrQ{\forall}{\traceVar_1}\ \constrQ{\forall}{\traceVar_2}\ \constrQ{\exists}{\traceVar'_1} \psi'\) where \(\psi'\) has only temporal quantifiers.
From \cite{satHyperLTL16,keysdecideHyperltl20}, it is undecidable to check for satisfiability of \(\HLTL\) formulas with quantifier pattern \(\forall\forall\exists\).
Hence, it is also not possible to decide \(\generatedSet{\varphi}\mathop{\neq}\emptyset\).
\end{proof}

\vspace{-2mm}
\subsection{Satisfiability}
\label{sec:trace:sat}
%\vspace{-2mm}

From Theorem \ref{thm:toUnconstrained} and Corollary \ref{cor:unconst:dec}, we get our first fragment of \(\TraceTFOL\) with a decidable satisfiability problem:
the fragment where constrained trace quantifiers is of shape \(\exists^* \forall^*\).
This matches the known decidable fragment for \(\HQPTL\) \cite{Hahn21}.

\begin{corollary}
\label{thm:trace_prefic_dec}
Let \(\varphi \mathop{=} \overrightarrow{\QExists_{\traceS}} \
\constrQ{\exists}{\traceVar_1} \LLL \constrQ{\exists}{\traceVar_n}\
\constrQ{\forall}{\traceVar'_1}\LLL \constrQ{\forall}{\traceVar'_m}\
\overrightarrow{\Quant_{\traceS}} \
\overrightarrow{\Quant_{\timeS}} \
\varphi_{\tfree}\)
be a hypertrace formula in prenex normal form s.t.\ 
\(\overrightarrow{\QExists_{\traceS}}\) is a sequence of existential unconstrained trace quantifiers, \(\overrightarrow{\Quant_{\traceS}}\) and 
\(\overrightarrow{\Quant_{\timeS}}\) are any combination of time and unconstrained trace quantifiers, respectively, and \(\varphi_{\tfree}\) is a quantifier-free.
It is decidable to check whether \(\generatedSet{\varphi}\neq \emptyset\).
\end{corollary}

Unconstrained trace quantifiers not only extend the expressivity of \(\TraceTFOL\) compared to \(\cTraceTFOL\) but can also be used to seamlessly define a semi-decision procedure to check for the unsatisfiability of \(\cTraceTFOL\) formulas.
In the proposition below, we prove that removing constraints from existentially quantified trace variables preserves the models of the constrained formula.
We can use the contrapositive of this proposition to determine the unsatisfiability of constrained hypertrace formulas.

\begin{proposition}
\label{prop:sat:constr_to_uncostr}
Let \(\varphi \mathop{=} \constrQ{\forall}{\traceVar_0}\ \constrQ{\exists}{\traceVar'_0}  \ldots \constrQ{\forall}{\traceVar_k}\ \constrQ{\exists}{\traceVar'_k}\ \psi\) be a hypertrace formula where \(\psi\) is quantifier-free, and \(\setTraces\) be a set of traces.
If \(\setTraces \modelsTFOL \varphi\), then \(\setTraces \modelsTFOL\constrQ{\forall}{\traceVar_0}\ \exists \traceVar'_0  \ldots \constrQ{\forall}{\traceVar_k}\ \exists \traceVar'_k\ \psi\).
\end{proposition}
\begin{proof}
The statement follows directly from \(\setTraces\) being a subset of all possible traces and, thus, any assignment over \(\setTraces\) is also an assignment over the set of all traces.
\end{proof}

\subsection{Equivalence to \(\HQPTL\)}
\label{sec:trace_hqptl}

We prove that, for sets of infinite traces, trace-prefixed hypertrace logic is expressively equivalent to \(\HQPTL\)~\cite{Rabe2016}.
\(\HQPTL\) formulas \(\varphi\) are defined by the grammar:
\begin{align*}
\psi ::= &  \ \Next \psi \ | \ \psi \Until \psi \ | \ \neg \psi \ |\ \psi \vee \psi \ | \ q \ |\ a_{\traceVar} \hspace{10mm}
\varphi ::=  \  \exists q \, \varphi \ |\ \forall q\, \varphi \ |\ \exists \traceVar\, \varphi \ | \ \forall \traceVar \, \varphi \ |\  \psi
\end{align*}
where \(\traceVar \mathop{\in} \VarTrace\) is a trace variable  and  \(q,a \mathop{\in} \Prop\) are propositional variables, where \(\VarTrace\mathop{\cap}\Prop\mathop{=}\emptyset\).
We evaluate  \(\HQPTL\) formulas over a set of traces and a trace assignment 
\(\traceAssign\As \VarTrace \rightarrow (\AllVals)^\omega\), relative to a time point $i\in\nat$:
\begin{align*}
%%% Quantification over propositions
    &(\traceAssign,{\setTraces}, i) \modelsHyper \exists  q \ \psi
	\tIff 
	\text{there exists } \trace \mathop{\in} (\Bool^{\{q\}})^\omega: (\traceAssign[\pi_q \mapsto \trace],{\setTraces}, i) \modelsHyper \psi;\\
	&(\traceAssign,{\setTraces}, i) \modelsHyper \forall q \ \psi 
	\tIff 
	\text{for all } \trace \mathop{\in}  (\Bool^{\{q\}})^\omega: (\traceAssign[\pi_q \mapsto \trace],{\setTraces}, i) \modelsHyper \psi;\\
%%% Quantification over traces
	&(\traceAssign, {\setTraces}, i) \modelsHyper \exists \pi \ \psi
	\tIff 
	\text{there exists } \trace \mathop{\in} \setTraces: (\traceAssign[\pi \mapsto \trace], {\setTraces}, i) \modelsHyper \psi;\\
	&(\traceAssign,{\setTraces}, i) \modelsHyper \forall \pi \ \psi 
	\tIff 
	\text{for all } \trace \mathop{\in} \setTraces: (\traceAssign[\pi \mapsto \trace],{\setTraces}, i) \modelsHyper \psi;\\
%%% Propositionas eval
    &(\traceAssign,{\setTraces}, i) \modelsHyper q 
	\tIff  \  
	q\mathop{\in}\traceAssign(\pi_q)[i];\hspace{40mm}
	(\traceAssign,{\setTraces}, i) \modelsHyper a_{\pi} 
	\tIff  \  
	a\mathop{\in}\traceAssign(\pi)[i];\\
%%% Time and Bools
	&(\traceAssign,{\setTraces}, i) \modelsHyper \neg  \psi 
	\tIff \ 
	(\traceAssign,{\setTraces}, i) \notmodelsHyper  \psi;\\
	&(\traceAssign,{\setTraces}, i) \modelsHyper  \psi_1 \vee  \psi_2 \tIff \ 
	(\traceAssign,{\setTraces}, i) \modelsHyper  \psi_1 \tOr (\traceAssign,{\setTraces}, i) \modelsHyper  \psi_2;\\
	&(\traceAssign,{\setTraces}, i) \modelsHyper \Next  \psi \tIff \  (\traceAssign,{\setTraces}, i+1) \modelsHyper  \psi;\\
	&(\traceAssign,{\setTraces}, i) \modelsHyper  \psi_1 \Until  \psi_2 \tIff 
	\text{exists } j \geq i:\! (\traceAssign,{\setTraces}, j) \modelsHyper  \psi_2 
    \text{ and  all } i \leq j' < j\ (\traceAssign,{\setTraces}, j') \modelsHyper  \psi_1.
\end{align*}
A set \(\setTraces\) of traces is a model of a \(\HQPTL\) formula \(\varphi\), denoted \(\setTraces \modelsHyper \varphi\),
iff there exists an assignment \(\traceAssign\) such that \((\traceAssign,{\setTraces},0) \modelsHyper \varphi\).
%
%A formula is closed (a sentence) when all occurrences of  trace variables are in the scope of a quantifier.
%
For all closed formulas  \(\varphi\), \(\setTraces \modelsHyper \varphi\) iff \((\emptyAssign, {\setTraces},0) \modelsHyper \varphi\), 
where \(\emptyAssign\) is the empty assignment.
%\xxx{Should we just write $\emptyset$ for the empty assignment?}
%
%We may omit the subscript \(H\) in \(\models_H\) when it is clear from context.

\begin{remark}
In this work, we interpret \(\HQPTL\) formulas as introduced in  \cite{Rabe2016,hierarchyHyper19}.
An alternative definition found in the literature uniformly instantiates quantified propositions across the set of traces used as the model \cite{Hahn21,hyperOmegaReg20,LTLTeam21,hyperqptlPluscomplexity24}.
The semantics we adopt in this work subsumes the uniform interpretation; that is, we can simulate the uniform interpretation by rewriting the \(\HQPTL\) formula interpreted under the uniform semantics.
\end{remark}

For a set of traces \(\setTraces\) and an assignment \(\traceAssign\) for
variables in \(\mathcal{V}\), we define its \emph{flattening} to a trace as:
\(a_{\traceVar}\mathop{\in}\flatT{\traceAssign}[i] \tIff a\mathop{\in}\traceAssign(\pi)[i]\)
and
\(q\mathop{\in}\flatT{\traceAssign}[i] \tIff q\mathop{\in}\traceAssign(\traceVar_q)[i]\)
for all \(\timeVar \mathop{\in}\nat\), trace variables \(\traceVar\in \Var\) and propositional variables \(a,q\mathop{\in} \Prop\).
A trace assignment satisfies a quantifier-free \(\HQPTL\) formula iff its flattening satisfies the same formula under the \(\LTL\) semantics.
% 
%We observe that the set of traces is not relevant anymore when we have only quantifier-free formulas.

\begin{proposition}\label{prop:zipping}
	Let \(\varphi\) a quantifier-free \(\HQPTL\) formula.
	For all  \(i\in\nat\), all trace sets \(\setTraces\), and all trace assignments \(\traceAssign\),
	\((\traceAssign, {\setTraces},i) \modelsHyper \varphi\) iff \(\flatT{\traceAssign}[i \ldots ] \modelsLTL \varphi.\)
\end{proposition}

\begin{theorem}
	\label{thm:tracepre_hqptl}
	For all \(\HQPTL\) sentences \(\varphi_H\) there exists a trace-prefixed hypertrace sentence \(\varphi\) such that
	for all sets of infinite traces \(\setTraces \mathop{\subseteq} (\AllVals)^\omega\),
	 \(\setTraces \modelsHyper \varphi_H\) iff \(\setTraces \modelsTFOL \varphi\).
	For all trace-prefixed hypertrace sentences \(\varphi\) there exists a \(\HQPTL\) sentence \(\varphi_H\) such that
	for all sets of infinite traces \(\setTraces \mathop{\subseteq} (\AllVals)^\omega\),
	 \(\setTraces \modelsHyper \varphi_H\) iff
	\(\setTraces  \modelsTFOL \varphi\).
\end{theorem}

\newcommand{\transl}{\mathtt{tr}}
\begin{proof}
We denote by \(\mathtt{LTLtoFO}(\psi)\) and \(\mathtt{FOtoLTL}(\psi)\) the translation from \(\LTL\) formulas to \(\FOLOrder\) and vice-versa given by 
 \(\LTL\) and \(\FOLOrder\) being equivalent \cite{gabbay1980temporal}.

We start with the translation from \(\HQPTL\) formulas to an equivalent \(\TraceTFOL\) formula. 
We first define the translation of \(\varphi\mathop{\in}\HQPTL\) for its different quantifiers:
\begin{equation}
\transl_H^{\text{quant}}(\varphi) =
\begin{cases}
 \varphi & \tIf \varphi \text{ is quantifier-free}\\
 \constrQ{\Quant}{\traceVar}\ \transl_H^{\text{quant}}(\varphi')  & \tIf \varphi\mathop{=} \Quant \traceVar\ \varphi', \Quant\mathop{\in}\{\forall,\exists\} \tAnd \traceVar\mathop{\in}\VarTrace\\
 \Quant \traceVar_q\ \transl_H^{\text{quant}}(\varphi') & \tIf \tIf \varphi\mathop{=} \Quant q\ \varphi', \Quant\mathop{\in}\{\forall,\exists\} \tAnd q\mathop{\in}\Prop
\end{cases}
\end{equation}
With \(\mathtt{LTLtoFO}(\psi)\) all propositional variables \(a_{\traceVar}\) in \(\psi\) are mapped to \(P_{a_{\traceVar}}(\timeVar)\), while all \(q\)
are mapped to \(P_{q}(\timeVar)\).
We define the substitution from these predicates to the binary predicates of \(\logic\):
$\sigma \mathop{=} \{P_{a_{\traceVar}}(i) \mapsto \BinaryP_a(\traceVar, \timeVar) \mid a\mathop{\in}\Prop,\traceVar\mathop{\in}\VarTrace , \timeVar\mathop{\in}\VarTime\} \cup
\{P_{q}(i) \mapsto \BinaryP_q(\traceVar_q, \timeVar) \mid q\mathop{\in}\Prop,\traceVar\mathop{\in}\VarTrace, \timeVar\mathop{\in}\VarTime\}
$.
The translation from \(\HQPTL\) formulas to \(\TraceTFOL\) is defined below, where \(\psi\mathop{\in}\LTL\),  \(x_i\mathop{\in}\VarTrace\mathop{\cup}\Prop\), for \(1\mathop{\leq} i\mathop{\leq}n\), and \(\Quant\mathop{\in}\{\forall, \exists\}\):
\begin{align}
\label{thm:tr_hqptl_hypertrae}
\hspace{-4mm}
\transl_H&(\Quant x_1\! \ldots\! \Quant x_n \ \psi) \mathop{=}
\transl_H^{\text{quant}}(\Quant x_1\! \ldots\! \Quant x_n \ ((\mathtt{LTLtoFO}(\psi))[\sigma])).
\end{align}
%
%We observe that our translation guarantees that each trace variable \(\traceVar_q\), where \(q\mathop{\in}\VarTrace\), is used only in the binary predicate \(\BinaryP_q\).
%
Given a trace assignment, \(\traceAssign\), and \(k\mathop{\in}\nat\) we denote \(\traceAssign_k\) any assignment satisfying
\(q\mathop{\in}\traceAssign_k(\traceVar_q)[i]\) iff \(q\mathop{\in}\traceAssign(\traceVar_q)[i+k]\), and otherwise, \(\traceAssign_k(\traceVar)\mathop{=}(\traceAssign(\traceVar))[k\ldots]\).
We prove in the extended version, using the equivalence between \(\LTL\) and \(\FOLOrder\) \cite{gabbay1980temporal} and Prop. \ref{prop:zipping}, that for all  \(\HQPTL\) formula \(\varphi\mathop{=}\Quant x_1 \ldots \Quant x_n \ \psi\) where
\(x_i\mathop{\in}\VarTrace\mathop{\cup}\Prop\), for \(1\mathop{\leq} i\mathop{\leq}n\), and \(\Quant\mathop{\in}\{\forall, \exists\}\):
\((\traceAssign,{\setTraces}, k)\modelsHyper \varphi\) iff
\((\overline{\setTraces[k\ldots]},(\traceAssign_k, \AssignN^{\emptyset}))\modelsTFOL \transl_{H}(\varphi)\), where \(\AssignN^{\emptyset}\) is the empty time assignment.

For the translation from trace-prefixed hypertrace formulas to \(\HQPTL\), we use Lemma~\ref{lemma:rewrite:independent}, and give a translation from flattened hypertrace formulas to \(\HQPTL\).
As before we define a substitution 
\(\sigma''\mathop{=}\{\BinaryP_a(\traceVar, \timeVar) \mapsto P_{a_{\traceVar}}(i) \mid a\mathop{\in}\Prop,\traceVar\mathop{\in}\VarTrace, \timeVar\mathop{\in}\VarTime\}\cup
\{\BinaryP_q(\traceVar_q, \timeVar) \mapsto  P_{q}(i) \mid q\mathop{\in}\Prop,\traceVar\mathop{\in}\VarTrace, \timeVar\mathop{\in}\VarTime\}\).
The translation is defined as:
%\begin{align}
%\label{thm:tr_hypertrace_hqptl}
%\hspace{-4mm}
\(\transl_{\traceS}(\varphi) \mathop{=}
\overrightarrow{\Quant} (\mathtt{FOtoLTL}(\psi'[\sigma'']))\) with 
\(\overrightarrow{\Quant} \psi'\mathop{=}\rewriteFlat(\varphi,\Prop,\emptyset)
\)
where \(\overrightarrow{\Quant}\) is a sequence of constrained and unconstrained trace quantifiers and \(\psi'\) is a hypertrace formula with no trace quantifiers.

We prove that for all hypertrace formulas \(\varphi\) with only free trace variables,
for all sets of traces \(\setTraces\) and assignments \(\AssignT\):
\((\overline{\setTraces},(\AssignT, \AssignN^{\emptyset}))\modelsTFOL \varphi\) iff
\((\AssignT',{\setTraces}, 0)\modelsHyper \transl_{\traceS}(\varphi)\), where
\(\AssignT'(\traceVar_q)=(\AssignT(\traceVar_q)[0]\mathop{\cap}\{q\}) (\AssignT(\traceVar_q)[1]\mathop{\cap}\{q\}) \ldots\).
For the base case, where \(\psi'\) has no trace quantifiers and all time quantifiers are bounded: 
\((\overline{\setTraces},(\AssignT, \AssignN^{\emptyset}))\modelsTFOL \psi'\) iff
\((\overline{\flatT{\AssignT}}, \AssignN^\emptyset) \models \psi'[\sigma']\), follows from an analogous proof from the previous case.
As the translation from \(\AssignT\) to \(\AssignT'\) does not change the valuation of \(q\) in the trace assigned to \(\traceVar_q\) and 
by \cite{gabbay1980temporal},
\((\overline{\flatT{\AssignT}}, \AssignN^\emptyset) \models \psi'[\sigma']\) iff
\(\flatT{\AssignT'}\modelsLTL\mathtt{FOtoLTL}(\psi'[\sigma'])\).
By Proposition \ref{prop:zipping}, it is equivalent to
\((\AssignT',{\setTraces}, 0)\modelsHyper \mathtt{FOtoLTL}(\psi'[\sigma'])\).
The induction cases (i.e., constrained and unconstrained quantifiers) follow from induction hypothesis and definitions.
\end{proof}

\section{Time-prefixed Hypertrace Logic}
\label{sec:time_prefix}

We consider now time-prefixed hypertrace logic, with formulas defined by the grammar:
\begin{align*}
    \varphi ::=  \exists \timeVar\ \varphi \ |\ \neg \varphi \ |\  \psi
	 \hspace{10mm}
     \psi &::= \exists \traceVar\, \psi \ |\  \constrQ{\exists}{\traceVar}\, \psi \ |\ \psi \vee \psi \ |\ \neg \psi \ |\  \timeVar < \timeVar\ |\  \timeVar = \timeVar\ |\   \BinaryP(\traceVar,\timeVar)
\end{align*}
where \(\traceVar\) is a trace variable,
\(\timeVar\) is a time variable and \(\BinaryP\) is a binary predicate.
%
%This fragment is orthogonal to both the unconstrained and the trace-prefixed fragments introduced before.
%
%\subsection{Relation to Time-prefixed Hyperlogics}
%\label{sec:time_pre:hyper}
%
To the best of our knowledge, there is no formalism in the literature allowing to specify hyperproperties by quantifying over time before traces, while supporting arbitrary time and trace quantifiers.
The {relative expressiveness} between the trace-prefix and time-prefix fragments, however, remains an open problem.
%

%  {\color{red} There is no hyperlogic that is equivalent to this fragment. There however fragments of other logics that can be captured by this fragment.}
 
% \ana{Example and justification for language. Teams semantics without splitjunction. 
% %
% Epistemic logic: compare traces with the same observations:

% %\(\forall i\ \neg (\exists j<i \ \constrQ{\exists}{\traceVar}\ \constrQ{\exists}{\traceVar'}\  \mathrm{obs}(\traceVar,j)=\mathrm{obs}(\traceVar',j) \wedge \mathrm{secret}(\traceVar,j)\neq\mathrm{secret}(\traceVar',j) )\)
% \(\exists i\ \forall j<i \ \constrQ{\forall}{\traceVar}\ \constrQ{\forall}{\traceVar'}\  \mathrm{obs}(\traceVar,j)=\mathrm{obs}(\traceVar',j) \rightarrow \mathrm{outcome}(\traceVar,i)=\mathrm{outcome}(\traceVar',i)\).
% }

%\subsection{Satisfiability}
%\label{sec:time_pre:sat}
\newcommand{\mem}{\mathtt{mem}}
\newcommand{\toC}{\mathtt{to}}
\newcommand{\inc}{\mathtt{inc}}
\newcommand{\dec}{\mathtt{dec}}
\newcommand{\isZ}{\mathtt{isZero}}
\newcommand{\guess}{\mathtt{guess}}
\newcommand{\Stop}{\mathtt{guessed}}

\medskip
As for the previously studied fragments, we use Theorem \ref{thm:toUnconstrained}, 
and Corollary \ref{cor:unconst:dec}
to identify 
a fragment of the time-prefixed hypertrace logic with decidable satisfiability problem.

\begin{corollary}
\label{thm:time_pre_dec}
Let \(\varphi \mathop{=} \overrightarrow{\QExists_{\timeS}} \
\constrQ{\exists}{\traceVar_1} \LLL \constrQ{\exists}{\traceVar_n}\
\constrQ{\forall}{\traceVar'_1}\LLL \constrQ{\forall}{\traceVar'_m}\
\overrightarrow{\Quant_{\traceS}} \
\varphi_{\tfree}\)
be a hypertrace formula in prenex normal form s.t.\ 
\(\overrightarrow{\QExists_{\timeS}}\) is a sequence of existential time quantifiers, \(\overrightarrow{\Quant_{\traceS}}\) is any combination of time and unconstrained trace quantifiers, respectively, and \(\varphi_{\tfree}\) is a quantifier-free.
It is decidable to check whether \(\generatedSet{\varphi}\neq \emptyset\).\del{.}
\end{corollary}

In our theorem below, we prove that the satisfiability problem for arbitrary time-prefixed formulas is undecidable.
We prove our result with a reduction from the (non)-halting problem for 2-counter Minsky machines.

A \emph{2-counter Minsky machine} is defined by a tuple \(\mathcal{M}\mathop{=}(Q,\Delta,\hat{q})\)
where \(Q\) is a finite set of states with \(\hat{q}\mathop{\in}Q\) being the initial state, 
and \(\Delta\mathop{\subseteq} Q \times \{1,2\} \times \{\inc,\dec,\isZ\} \times Q\) being the transition relation.
For all \(n,n'\mathop{\in}\nat\), we write:
\(n \xrightarrow{\isZ}n'\) iff \(n\mathop{=}n'\mathop{=}0\);
\(n \xrightarrow{\inc}n'\) iff \(n'\mathop{=}n+1\); and
\(n \xrightarrow{\dec}n'\) iff \(n'\mathop{=}n-1\).
We observe that \(\dec\) is only allowed when \(n>0\).
A configuration is a tuple with a state and the current value of the two counters.
Given one of the counters \(c\mathop{\in}\{1,2\}\) we refer to the \del{the} other counter as
\(\overline{c}\mathop{=}3-c\).
We say that there is a transition between two configurations \((q,n_1,n_2)\) and 
\((q',n'_1,n'_2)\) iff there exists \((q,c,op,q')\mathop{\in}\Delta\) s.t.\ 
\(n_c \xrightarrow{op}n'_c\) and, for the other counter, \(n_{\overline{c}}\mathop{=}n'_{\overline{c}}\).
A computation is a sequence of configurations connected by transitions.
It is undecidable to check whether an arbitrary 2-counter Minsky machine has an infinite computation \cite{minsky67}.

\begin{theorem}
\label{thm:time_prefix_undec}
Let \(\Quant\mathop{\in}\exists_{\timeS}
\forall_{\timeS}
\exists_{\timeS}^2
\forall_{\timeS}
\constrQ{\forall_{\traceS}}{}
(\constrQ{\exists_{\traceS}}{})^2
\exists_{\traceS}\) 
and \(\varphi\mathop{=}\Quant \varphi'\)
be a time-prefixed hypertrace formula where \(\varphi'\)
is quantifier-free.
It is undecidable to check whether \(\generatedSet{\varphi}\mathop{\neq}\emptyset\).
\end{theorem}

\begin{proof}
In our reduction from the (non)-halting problem for 2-counter Minsky machines, {each trace encodes one configuration} and the transition relation for the next step of the computation.
Traces include the following propositional variables to capture configurations, for a given set of states \(Q\):
\(\Prop_{\text{config}}\mathop{=} Q \mathop{\cup} \{\mem_1, \mem_2\}\) with \(\mem_1\) and \(\mem_2\) encoding the current value in their respective counters.
%
%\del{The values in each counter are represented in unary and so, incrementing the counter \(c\) accounts to making \(\mem_i\) true for one more step, while decrementing removes the last true valuation.}
%
We use
\(\Prop_{\text{transition}}\mathop{=} \{q_{\text{to}}, q_{\text{from}}\,|\, q\mathop{\in}Q\} \mathop{\cup} \{\inc, \dec,\isZ, \toC_1, \toC_2 \}\) to specify transitions
where \(\toC_1\) and \(\toC_2\) represent the counters to be updated.
The values in each counter are represented in unary and so incrementing the counter \(c\) {amounts} to making \(\mem_i\) true for one more step, while decrementing removes the last true valuation.
We \del{will} use an unconstrained trace quantifier to guess the time point where this increment or decrement takes place.
This guess is guided by two propositional variables:
\(\Prop_{\text{guess}}\mathop{=}\{\guess, \Stop\}\).
Hence, our traces are defined over
\(\Prop\mathop{=}\Prop_{\text{config}} \mathop{\cup} \Prop_{\text{transition}}  \mathop{\cup} \Prop_{\text{guess}}\).

\newcommand{\exactO}{\mathtt{exactlyOne}}
We start by defining useful formulas to encode requirements of well-formed traces, where \(\exactO(\PropY)\) is true when only one of the propositions in \(\PropY\) holds:
\begin{itemize}
\item At each time \(i\) the trace \(\traceVar\) defines an unique transition:
%\xxx{Maybe better name for uniqeTr would be singleTr? Or detTr (like deterministic?)}
%%%%%%%%%%%%%%%
    {\small
    \begin{align}
    \hspace{-7mm}
     \mathtt{singleTr}&(Q,\traceVar,i) \Def 
    \exactO(\{q_{\text{to}}(\traceVar,i)\mid q\mathop{\in}Q\})\ \wedge 
    \exactO(\{q_{\text{from}}(\traceVar,i) \mid q\mathop{\in}Q\}) \wedge\\ 
    &\exactO(\{\toC_{1}(\traceVar,i), \toC_2(\traceVar,i)\})\wedge 
    \exactO(\{\inc(\traceVar,i) ,\dec(\traceVar,i),\isZ(\traceVar,i)\})\nonumber
    \end{align}
    }
%%%%%%%%%%%%%%%

\item for times \(i\) and \(j\) the trace \(\traceVar\) defines the same transition:
%%%%%%%%%%%%%%%
    {\small
    \begin{align}
    \hspace{-7mm}
    \mathtt{sameTr}(Q,\traceVar,i,j)\Def &
    \bigwedge\limits_{q\mathop{\in}Q} 
    (q_{\text{to}}(\traceVar,i) \leftrightarrow q_{\text{to}}(\traceVar,j)) \wedge 
    \bigwedge\limits_{q\mathop{\in}Q} 
    (q_{\text{from}}(\traceVar,i) \leftrightarrow q_{\text{from}}(\traceVar,j))\ \wedge \\
    & \bigwedge\limits_{c\mathop{\in}\{1,2\}} 
    (\toC_{c}(\traceVar,i) \leftrightarrow \toC_{c}(\traceVar,j)) \wedge 
    \!\!\!\! \bigwedge\limits_{op\mathop{\in}\{\inc,\dec,\isZ\}} \!\!\!\!\!\!
    (op(\traceVar,i) \leftrightarrow op(\traceVar,j)) \nonumber
    \end{align}
    }  
    % \xxx{Shouldn't the disjunctions be conjunctions?}
    % \xxx{Ana: It is encoding - there exists one value that is the same. This together with uniqueness gives that it is the same transition.}
    % \xxx{Conjunction would work too, encoding that all transitions are the same (and then uniqueness would tell there is only one}
%%%%%%%%%%%%%%%
\item transition matches a transition from \(\Delta\):
%\anaSide{\textbf{R1:}  swap to and from.}
%%%%%%%%%%%%%%%
    {\small
    \begin{align}
    \hspace{-7mm}
    \mathtt{validTr}(\Delta,\traceVar,i)\Def &
    \bigvee\limits_{(q,c,op,q')\mathop{\in}\Delta} (q_{{\text{from}}}(\traceVar,i)) \wedge \toC_c(\traceVar,i)) \wedge op(\traceVar,i)) \wedge q'_{{\text{to}}}(\traceVar,i)))
    \end{align}
    }  
%%%%%%%%%%%%%%%
\item for times \(i\) and \(j\) the trace \(\traceVar\) defines the same unique state:
%%%%%%%%%%%%%%%
{\small
\begin{align}
\hspace{-7mm}
\mathtt{sameState}(Q,\traceVar,i,j) &\Def (\bigwedge\limits_{q\mathop{\in}Q}q(\traceVar,i) \leftrightarrow q(\traceVar,j)) \wedge 
\exactO(\{q(\traceVar,i)\mid q\mathop{\in}Q\})
\end{align}
}
%%%%%%%%%%%%%%%
\item states in \(\traceVar\) and \(\traceVar'\) match the configuration:
%%%%%%%%%%%%%%%
{\small
\begin{align}
\hspace{-7mm}
\mathtt{goodStates}(Q,\traceVar,\traceVar',i) &\Def \bigvee_{q\mathop{\in}Q} (q_{\text{from}}(\traceVar, i) \leftrightarrow q(\traceVar,i)) \wedge \bigvee_{q\mathop{\in}Q} (q_{\text{to}}(\traceVar, i) \leftrightarrow q(\traceVar',i))
\end{align}
}
%\xxx{We could simplify the formula and text by merging sameTr and sameState, having sameProps(Y) for the set of propositions Y.}
%%%%%%%%%%%%%%%
% \item encoding of initial state:
% {\small
% \begin{align}
% \hspace{-7mm}
% \mathtt{isInitial}(\hat{q},\traceVar,i) &\Def \hat{q}(\traceVar,i) \wedge (\neg \mem_1(\traceVar,i)) \wedge (\neg \mem_2(\traceVar,i))
% \end{align}
% }
%%%%%%%%%%%%%%%
\item Only one (the first) guess can affect the evaluation:
{\small
\begin{align}
\hspace{-7mm}
\mathtt{stopMulitpleGuess}(\traceVar_g,i,j) &\Def (\guess(\traceVar_g,i)\vee \Stop(\traceVar_g,i)) \rightarrow \Stop(\traceVar_g,j)
\end{align}
}
\end{itemize}

\newcommand{\iSucc}{i_{+}}
\newcommand{\iPre}{i_{-}}
We recall that operations update at most one of the counters.
To guess for the time point where this change occurs, our encoding uses propositional variables \(\guess\) and \(\Stop\) to be instantiated by an unconstrained trace quantifier.
In particular, if either the \(\guess\) is false or \(\Stop\) is true, we require the two traces being compared to have equal values in their counters.
We define now how to encode the effect of an operation, where \(\traceVar\) encodes the transition to be applied as well as the incoming configuration,  \(\traceVar'\) encodes the outgoing configuration and \(\traceVar_g\) the unconstrained trace with the guess:
\vspace{-1mm}
{\small
    \begin{align}
    \hspace{-4.3mm}
    \mathtt{op}(&\traceVar,c,\traceVar',\traceVar_g, i, \iSucc, \iPre) \Def
    %% Other counter should not change
    (\mem_{\overline{c}}(\traceVar,i) \leftrightarrow \mem_{\overline{c}}(\traceVar',i)\ \wedge\\
    %% If the op=isZero both counters should be zero
    &(\isZ(\traceVar,i) \rightarrow (\neg \mem_c(\traceVar,i) \wedge \neg \mem_c(\traceVar',i)))\ \wedge \\
    %% If the guess is zero both counters should be equal
    &((\neg \guess(\traceVar_g,i) \vee \Stop(\traceVar_g,i)) \rightarrow (\mem_c(\traceVar,i) \leftrightarrow \mem_c(\traceVar',i)))\ \wedge \\
    %% Op is increment
    &
    ((\neg \Stop(\traceVar_g,i) \wedge \mathtt{guess}(\traceVar_g,i) \wedge \inc(\traceVar,i)) \rightarrow
    (\neg \mem_c(\traceVar,i) \wedge \mem_c(\traceVar',i) \wedge \neg \mem_c(\traceVar',\iSucc)) \wedge \\
    %% Op is decrement
    &
    ((\neg \Stop(\traceVar_g,i) \wedge \mathtt{guess}(\traceVar_g,i) \wedge \dec(\traceVar,i)) \rightarrow
    ( \mem_c(\traceVar,i) \wedge \neg \mem_c(\traceVar',i) \wedge \mem_c(\traceVar',\iPre)) 
    \end{align}
}
We include the time just before \(i\), \(\iPre\), and just after, \(\iSucc\), because time quantification will precede the quantification of \(\traceVar'\) and, so, we may get a different \(\traceVar'\) at each time point.
In the time-prefixed formula defined next, we encode the requirement that all true valuations of \(\mem_1\) and \(\mem_2\) must be consecutive, starting from the beginning of the trace.
Hence, by using \(\iPre\) and \(\iSucc\) we ensure that the change around time \(i\) is consistent with the operation leading to it.
We combine all requirements above in the following time-prefixed hypertrace logic formula \(\varphi_{\mathcal{M}}\) for a given machine \(\mathcal{M}\mathop{=}(Q,\Delta,\hat{q})\):
%\xxx{Shouldn't $\exists \pi_{\hat q}$ be placed before $\forall \pi::T$?}
%\xxx{$i_0$ is supposed to be $0$?}
\vspace{-1mm}
{\small
\begin{align}
\hspace{-4mm}
\exists i_0 
\forall i\ \exists \iSucc \ &\exists \iPre\ \forall j\ \constrQ{\forall}{\traceVar}\ \constrQ{\exists}{\traceVar'}\ \constrQ{\exists}{\traceVar_{\hat{q}}}\ \exists \traceVar_{g}
\big(\\
&\hspace{-11mm}(i_0\leq j \wedge \iPre\leq i \wedge i < \iSucc \wedge (i<j \rightarrow \iSucc\leq j) \wedge (\iPre=i \rightarrow i=i_0) \wedge (j<i \rightarrow j\leq \iPre)) \rightarrow
\\ 
&
(\mem_1(\traceVar,\iSucc)\rightarrow\mem_1(\traceVar,i)) \wedge (\mem_2(\traceVar,\iSucc)\rightarrow\mem_2(\traceVar,i)) \ \wedge
\\
&\hat{q}(\traceVar_{\hat{q}},i) \wedge (\hat{q}(\traceVar,i) \rightarrow (\neg \mem_1(\traceVar,i) \wedge \neg \mem_2(\traceVar,i)))\ \wedge\\ 
&\mathtt{stopMulitpleGuess}(\traceVar_g,i,\iSucc)\ \wedge
\\
&\mathtt{singleTr}(Q,\traceVar,i) \wedge 
\mathtt{sameTr}(Q,\traceVar,i,\iSucc) \wedge 
\mathtt{validTr}(\Delta,\traceVar,i)\ \wedge 
\\
&
\mathtt{sameState}(Q,\traceVar,i,\iSucc)  \wedge
\mathtt{goodStates}(Q,\traceVar,\traceVar',i)\ \wedge\\
&(\toC_{1}(\traceVar,i) \rightarrow \mathtt{op}(\traceVar,1, \traceVar', \traceVar_g, i, \iSucc, \iPre)) \wedge (\toC_{2}(\traceVar,i) \rightarrow \mathtt{op}(\traceVar,2, \traceVar', \traceVar_g, i, \iSucc, \iPre ))
 \big)
\end{align}
}
A set of traces \(\setTraces\modelsTFOL \varphi_{\mathcal{M}}\) satisfies the following properties, where for \(\PropY\mathop{\subseteq}\Prop\) we define its projection over a trace as \(\trace|_{\PropY}=(\trace[0]\mathop{\cap}\PropY)(\trace[1]\mathop{\cap}\PropY)\ldots\):
\begin{itemize}
    \item exists \(\trace\mathop{\in}\setTraces\) and \((\hat{q},c,op,q')\mathop{\in}\Delta\) s.t.\ \(\trace|_{\{\hat{q},\hat{q}_{\text{from}}, \toC_c, op, q'_{\text{to}},\mem_1, \mem_2\}} \mathop{=} \{\hat{q},\hat{q}_{\text{from}}, \toC_c, op, q'_{\text{to}}\}^{\omega}\) encoding the initial state where both counters are \(0\);
    %\anaSide{\new{From (24), (26) and (28)}}
    %
    \item for all \(\trace\mathop{\in}\setTraces\), \(\trace|_{\{\mem_1\}}\mathop{=} \{\mem_{1}\}^{n_1} \{ \}^\omega\), \(\trace|_{\{\mem_2\}}\mathop{=} \{\mem_{2}\}^{n_2} \{ \}^\omega\)
    %\anaSide{\new{From (24) and (25)}.}
    and there exists \(q\mathop{\in}Q\) and \((q,c,op,q')\mathop{\in}\Delta\) s.t.\ 
    \(\trace|_{\{q,q_{\text{from}},\toC_c,op,q'_{\text{to}}\}}=\{q,q_{\text{from}},\toC_c,op,q'_{\text{to}}\}^{\omega}\), which are unique
   % \anaSide{\new{From (24),(28) and (29).}}
    and denoted by 
    \(\text{config}(\trace){=}(q,n_1,n_2)\) and \(\text{trans}(\trace)\mathop{=}(q,c,op,q')\);
    \item for each trace assigned to \(\traceVar\), there can be only one time point where for the unconstrained trace assigned to \(\traceVar_g\) the value of \(\guess\) is true and \(\Stop\) is false;
   % \anaSide{\new{From (27).}}
    %
    \item the trace assigned to \(\traceVar'\) encodes a correct next state from the transition defined by \(\traceVar\).
   % \anaSide{\new{From (24), (27), (28).}}
\end{itemize}
\medskip
%\anaSide{Expand this!}
\(\setTraces\modelsTFOL\varphi_{\mathcal{M}}\), {iff} 
\(\setTraces\) has a
trace encoding the initial state of \(\mathcal{M}\) and, 
{starting from that trace we can  define a valid infinite computation of \(\mathcal{M}\) using only traces in \(\setTraces\)}.
%
%\new{The \(\Rightarrow\)-direction follows from the points mentioned above.}
%\new{The \(\Leftarrow\)-direction we prove by contraposition. 
%
%Consider arbitrary set of traces s.t.\ \(\setTraces\notmodelsTFOL\varphi_{\mathcal{M}}\).
%}
\end{proof}

\vspace{-4mm}
\section{Related Work}
\vspace{-2mm}
\label{sec:rel_work}

The first work to explore the connection between hyperlogics (specifically \(\HLTL\)) and classical first-order reasoning is by Finkbeiner and Zimmermann~\cite{hyperFO17}.
They extend \(\FOLOrder\) with the \emph{equal-level} predicate \(E\), which allows comparisons between positions at the same time point across different traces. 
They also use the bounded-promptness (also referred to as bounded-termination) property from~\cite{bozzelli2015unifying} to show that \(FO[<,E]\) is strictly more expressive than HyperLTL.
Later, Bartocci et al. introduce hypertrace logic \cite{bartocci2022flavors}, which takes a different approach by extending \(\FOLOrder\) to a two-sorted logic, explicitly distinguishing between sorts for time and (constrained) traces.
They prove that \(FO[<,E]\) and hypertrace logic (with only constrained trace quantifiers) are equivalent.
%
%Hypertrace logic, however, allows for shorter specifications (i.e., properties to be written more succinctly than in \(FO[<,E]\)) and is particularly well-suited for analyzing the interaction between different patterns of time and trace quantification, as these are made explicit in the \(\logic\)’s syntax.
%
More recently, Beutner and Finkbeiner, in \cite{satHyperFOL25}, presented a tool to check the satisfiability of hyperproperties specified in \(\HLTL\) by translating them to first-order logic.
Their translation follows the same approach as hypertrace logic: it considers separate sorts for traces and time.
Other works have explore reasoning about hyperproperties using second-order quantification. 
For example, \(\HTwoLTL\)~\cite{secondHLTL23} extends \(\HLTL\) by introducing second-order quantification over the set of all possible traces, in addition to first-order quantification over a given trace set. 
These approaches fall outside the scope of our work, as we focus on first-order definable languages.

The boundaries for the decidability of the satisfiability problem of \(\HLTL\) were first investigated in \cite{satHyperLTL16}.
These results were extended, in \cite{keysdecideHyperltl20}, to understand further which \(\HLTL\) fragments have a decidable satisfiability checking.
In \cite{hierarchyHyper19}, Coenen et al. give a complete picture of the hierarchy of hyperlogics based on extensions of temporal logics.
In particular, they compare and summarize results for \(\HLTL\) and \(\HQPTL\).
A different line of work looks into \(\LTL\) interpretation with team semantics, called \(\TeamLTL\).
Team semantics generalizes classical first-order logic by evaluating formulas over sets of assignments, called teams, rather than single assignments.
\(\TeamLTL\),
based on this framework, does not use explicit trace quantifiers. 
Instead, it adopts the modular approach of team semantics, where reasoning over multiple traces is introduced implicitly through the use of generalized atoms.
In \cite{LTLTeam21}, the authors establish the relation between different extensions of \(\TeamLTL\) to fragments of \(\HQPTL^+\), which generalizes \(\HQPTL\) to include non-uniform quantification over propositional variables.

\vspace{-2mm}
\section{Conclusion}
\vspace{-2mm}
In this work, we introduced an extension of hypertrace logic with unconstrained trace quantifiers, effectively extending its expressive power. 
We established expressiveness results connecting fragments of hypertrace logic to well-known temporal logics: unconstrained hypertrace logic is equivalent to \(\SOneS\), while trace-prefixed hypertrace logic corresponds to \(\HQPTL\).
We also identified a decidable fragment of the logic, where constrained trace quantifier alternation is restricted to a single existential-to-universal switch. 
This generalizes a known decidability result for \(\HQPTL\) by allowing arbitrary placements of time quantifiers. 
Finally, we identified a fragment of time-prefixed hypertrace logic with undecidable satisfiability checking.
For future work, it remains open how the trace- and time-prefixed fragments relate in terms of expressiveness. 
Additionally, we plan to explore the connection between hyperlogics and classical first-order reasoning to enable the transfer of techniques and results across these areas.

\bibliography{main}

\begin{thebibliography}{10}

\bibitem{bartocci2022flavors}
Ezio Bartocci, Thomas Ferr{\`e}re, Thomas~A. Henzinger, Dejan Nickovic, and
  Ana~Oliveira da~Costa.
\newblock Flavors of sequential information flow.
\newblock In Bernd Finkbeiner and Thomas Wies, editors, {\em Verification,
  Model Checking, and Abstract Interpretation (VMCAI)}, pages 1--19. Springer
  International Publishing, 2022.
\newblock \href {https://doi.org/10.1007/978-3-030-94583-1\_1}
  {\path{doi:10.1007/978-3-030-94583-1\_1}}.

\bibitem{satHyperFOL25}
Raven Beutner and Bernd Finkbeiner.
\newblock Checking satisfiability of hyperproperties using first-order logic.
\newblock In S.~Akshay, Aina Niemetz, and Sriram Sankaranarayanan, editors,
  {\em Automated Technology for Verification and Analysis (ATVA)}, pages
  198--211, Cham, 2025. Springer Nature Switzerland.
\newblock \href {https://doi.org/10.1007/978-3-031-78750-8_10}
  {\path{doi:10.1007/978-3-031-78750-8_10}}.

\bibitem{secondHLTL23}
Raven Beutner, Bernd Finkbeiner, Hadar Frenkel, and Niklas Metzger.
\newblock Second-order hyperproperties.
\newblock In Constantin Enea and Akash Lal, editors, {\em Computer Aided
  Verification (CAV)}, pages 309--332, Cham, 2023. Springer Nature Switzerland.
\newblock \href {https://doi.org/10.1007/978-3-031-37703-7_15}
  {\path{doi:10.1007/978-3-031-37703-7_15}}.

\bibitem{bozzelli2015unifying}
Laura Bozzelli, Bastien Maubert, and Sophie Pinchinat.
\newblock Unifying hyper and epistemic temporal logics.
\newblock In {\em International Conference on Foundations of Software Science
  and Computation Structures}, pages 167--182. Springer, 2015.

\bibitem{hierarchyHyper19}
Norine Coenen, Bernd Finkbeiner, Christopher Hahn, and Jana Hofmann.
\newblock The hierarchy of hyperlogics.
\newblock In {\em 2019 34th Annual ACM/IEEE Symposium on Logic in Computer
  Science (LICS)}, pages 1--13, 2019.
\newblock \href {https://doi.org/10.1109/LICS.2019.8785713}
  {\path{doi:10.1109/LICS.2019.8785713}}.

\bibitem{satHyperLTL16}
Bernd Finkbeiner and Christopher Hahn.
\newblock {Deciding Hyperproperties}.
\newblock In Jos\'{e}e Desharnais and Radha Jagadeesan, editors, {\em 27th
  International Conference on Concurrency Theory (CONCUR 2016)}, volume~59 of
  {\em Leibniz International Proceedings in Informatics (LIPIcs)}, pages
  13:1--13:14, Dagstuhl, Germany, 2016. Schloss Dagstuhl -- Leibniz-Zentrum
  f{\"u}r Informatik.
\newblock \href {https://doi.org/10.4230/LIPIcs.CONCUR.2016.13}
  {\path{doi:10.4230/LIPIcs.CONCUR.2016.13}}.

\bibitem{hyperOmegaReg20}
Bernd Finkbeiner, Christopher Hahn, Jana Hofmann, and Leander Tentrup.
\newblock Realizing $\omega$-regular hyperproperties.
\newblock In Shuvendu~K. Lahiri and Chao Wang, editors, {\em Computer Aided
  Verification (CAV)}, pages 40--63, Cham, 2020. Springer International
  Publishing.

\bibitem{hyperFO17}
Bernd Finkbeiner and Martin Zimmermann.
\newblock {The First-Order Logic of Hyperproperties}.
\newblock In Heribert Vollmer and Brigitte Vall\'{e}e, editors, {\em 34th
  Symposium on Theoretical Aspects of Computer Science (STACS 2017)}, volume~66
  of {\em Leibniz International Proceedings in Informatics (LIPIcs)}, pages
  30:1--30:14, Dagstuhl, Germany, 2017. Schloss Dagstuhl -- Leibniz-Zentrum
  f{\"u}r Informatik.
\newblock \href {https://doi.org/10.4230/LIPIcs.STACS.2017.30}
  {\path{doi:10.4230/LIPIcs.STACS.2017.30}}.

\bibitem{gabbay1980temporal}
Dov Gabbay, Amir Pnueli, Saharon Shelah, and Jonathan Stavi.
\newblock On the temporal analysis of fairness.
\newblock In {\em Proceedings of the 7th ACM SIGPLAN-SIGACT symposium on
  Principles of programming languages}, pages 163--173, 1980.

\bibitem{Hahn21}
Christopher Hahn.
\newblock {\em Logical and deep learning methods for temporal reasoning}.
\newblock 2021.
\newblock \href {https://doi.org/http://dx.doi.org/10.22028/D291-35192}
  {\path{doi:http://dx.doi.org/10.22028/D291-35192}}.

\bibitem{LinearHerlihyWing90}
Maurice~P. Herlihy and Jeannette~M. Wing.
\newblock Linearizability: a correctness condition for concurrent objects.
\newblock {\em ACM Trans. Program. Lang. Syst.}, 12(3):463–492, July 1990.
\newblock \href {https://doi.org/10.1145/78969.78972}
  {\path{doi:10.1145/78969.78972}}.

\bibitem{Kamp1968}
Hans Kamp.
\newblock {\em Tense Logic and the Theory of Linear Order}.
\newblock PhD thesis, {UCLA}, 1968.

\bibitem{qptl21202}
Yonit Kesten and Amir Pnueli.
\newblock Complete proof system for qptl.
\newblock {\em Journal of Logic and Computation}, 12(5):701--745, 10 2002.
\newblock \href {https://doi.org/10.1093/logcom/12.5.701}
  {\path{doi:10.1093/logcom/12.5.701}}.

\bibitem{keysdecideHyperltl20}
Corto Mascle and Martin Zimmermann.
\newblock {The Keys to Decidable HyperLTL Satisfiability: Small Models or Very
  Simple Formulas}.
\newblock In Maribel Fern\'{a}ndez and Anca Muscholl, editors, {\em 28th EACSL
  Annual Conference on Computer Science Logic (CSL 2020)}, volume 152 of {\em
  Leibniz International Proceedings in Informatics (LIPIcs)}, pages
  29:1--29:16, Dagstuhl, Germany, 2020. Schloss Dagstuhl -- Leibniz-Zentrum
  f{\"u}r Informatik.
\newblock \href {https://doi.org/10.4230/LIPIcs.CSL.2020.29}
  {\path{doi:10.4230/LIPIcs.CSL.2020.29}}.

\bibitem{minsky67}
Marvin~L. Minsky.
\newblock {\em Computation: finite and infinite machines}.
\newblock Prentice-Hall, Inc., USA, 1967.

\bibitem{Pnueli77}
Amir Pnueli.
\newblock The temporal logic of programs.
\newblock In {\em Proc. of FOCS77: the 18th Annual Symposium on Foundations of
  Computer Science}, pages 46--57. {IEEE} Computer Society, 1977.
\newblock \href {https://doi.org/10.1109/SFCS.1977.32}
  {\path{doi:10.1109/SFCS.1977.32}}.

\bibitem{Rabe2016}
Markus~N. Rabe.
\newblock {\em A temporal logic approach to information-flow control}.
\newblock 2016.
\newblock \href {https://doi.org/http://dx.doi.org/10.22028/D291-26650}
  {\path{doi:http://dx.doi.org/10.22028/D291-26650}}.

\bibitem{hyperqptlPluscomplexity24}
Ga{\"e}tan Regaud and Martin Zimmermann.
\newblock The complexity of hyperqptl.
\newblock {\em arXiv preprint arXiv:2412.07341}, 2024.

\bibitem{S1SBuchi60}
J.~{Richard Büchi}.
\newblock Symposium on decision problems: On a decision method in restricted
  second order arithmetic.
\newblock In Ernest Nagel, Patrick Suppes, and Alfred Tarski, editors, {\em
  Logic, Methodology and Philosophy of Science}, volume~44 of {\em Studies in
  Logic and the Foundations of Mathematics}, pages 1--11. Elsevier, 1960.
\newblock \href {https://doi.org/https://doi.org/10.1016/S0049-237X(09)70564-6}
  {\path{doi:https://doi.org/10.1016/S0049-237X(09)70564-6}}.

\bibitem{LTLTeam21}
Jonni Virtema, Jana Hofmann, Bernd Finkbeiner, Juha Kontinen, and Fan Yang.
\newblock {Linear-Time Temporal Logic with Team Semantics: Expressivity and
  Complexity}.
\newblock In Miko{\l}aj Boja\'{n}czyk and Chandra Chekuri, editors, {\em 41st
  IARCS Annual Conference on Foundations of Software Technology and Theoretical
  Computer Science (FSTTCS 2021)}, volume 213 of {\em Leibniz International
  Proceedings in Informatics (LIPIcs)}, pages 52:1--52:17, Dagstuhl, Germany,
  2021. Schloss Dagstuhl -- Leibniz-Zentrum f{\"u}r Informatik.
\newblock \href {https://doi.org/10.4230/LIPIcs.FSTTCS.2021.52}
  {\path{doi:10.4230/LIPIcs.FSTTCS.2021.52}}.

\end{thebibliography}

\section{Full Proofs}
\noindent{\bf Theorem \ref{thm:tracepre_hqptl}.}
\emph{
	For all \(\HQPTL\) sentences \(\varphi_H\) there exists a trace-prefixed hypertrace sentence \(\varphi\) such that
	for all sets of infinite traces \(\setTraces \mathop{\subseteq} (\AllVals)^\omega\),
	 \(\setTraces \modelsHyper \varphi_H\) iff \(\setTraces \modelsTFOL \varphi\).
	For all trace-prefixed hypertrace sentences \(\varphi\) there exists a \(\HQPTL\) sentence \(\varphi_H\) such that
	for all sets of infinite traces \(\setTraces \mathop{\subseteq} (\AllVals)^\omega\),
	 \(\setTraces \modelsHyper \varphi_H\) iff
	\(\setTraces  \modelsQPTL \varphi\).}

\begin{proof}
We denote by \(\mathtt{LTLtoFO}(\psi)\) and \(\mathtt{FOtoLTL}(\psi)\) the translation from \(\LTL\) formulas to \(\FOLOrder\) and vice-versa given by 
 \(\LTL\) and \(\FOLOrder\) being equivalent \cite{gabbay1980temporal}.

We start with the translation from \(\HQPTL\) formulas to an equivalent trace-prefixed hypertrace formula. 
We first define the translation of \(\varphi\mathop{\in}\HQPTL\) for its different quantifiers:
\begin{equation}
\transl_H^{\text{quant}}(\varphi) =
\begin{cases}
 \varphi & \tIf \varphi \text{ is quantifier-free}\\
 \constrQ{\Quant}{\traceVar}\ \transl_H^{\text{quant}}(\varphi')  & \tIf \varphi\mathop{=} \Quant \traceVar\ \varphi', \Quant\mathop{\in}\{\forall,\exists\} \tAnd \traceVar\mathop{\in}\VarTrace\\
 \Quant \traceVar_q\ \transl_H^{\text{quant}}(\varphi') & \tIf \tIf \varphi\mathop{=} \Quant q\ \varphi', \Quant\mathop{\in}\{\forall,\exists\} \tAnd q\mathop{\in}\Prop
\end{cases}
\end{equation}
With \(\mathtt{LTLtoFO}(\psi)\) all propositional variables \(a_{\traceVar}\) in \(\psi\) are mapped to \(P_{a_{\traceVar}}(\timeVar)\), while all \(q\)
are mapped to \(P_{q}(\timeVar)\).
We define the substitution from these predicates to the binary predicates of \(\logic\):
$\sigma \mathop{=} \{P_{a_{\traceVar}}(i) \mapsto \BinaryP_a(\traceVar, \timeVar) \mid a\mathop{\in}\Prop,\traceVar\mathop{\in}\VarTrace , \timeVar\mathop{\in}\VarTime\} \cup
\{P_{q}(i) \mapsto \BinaryP_q(\traceVar_q, \timeVar) \mid q\mathop{\in}\Prop,\traceVar\mathop{\in}\VarTrace, \timeVar\mathop{\in}\VarTime\}
$.
The translation from \(\HQPTL\) formulas to \(\TraceTFOL\) is defined below, where \(\psi\mathop{\in}\LTL\),  \(x_i\mathop{\in}\VarTrace\mathop{\cup}\Prop\), for \(1\mathop{\leq} i\mathop{\leq}n\), and \(\Quant\mathop{\in}\{\forall, \exists\}\):
\begin{align}
\label{thm:tr_hqptl_hypertrae}
\hspace{-4mm}
\transl_H&(\Quant x_1\! \ldots\! \Quant x_n \ \psi) \mathop{=}
\transl_H^{\text{quant}}(\Quant x_1\! \ldots\! \Quant x_n \ ((\mathtt{LTLtoFO}(\psi))[\sigma])).
\end{align}
%
%We observe that our translation guarantees that each trace variable \(\traceVar_q\), where \(q\mathop{\in}\VarTrace\), is used only in the binary predicate \(\BinaryP_q\).
%
Given a trace assignment, \(\traceAssign\), and \(k\mathop{\in}\nat\) we denote \(\traceAssign_k\) any assignment satisfying
\(q\mathop{\in}\traceAssign_k(\traceVar_q)[i]\) iff \(q\mathop{\in}\traceAssign(\traceVar_q)[i+k]\), and otherwise, \(\traceAssign_k(\traceVar)\mathop{=}(\traceAssign(\traceVar))[k\ldots]\).
We want to prove that
given a \(\HQPTL\) formula \(\varphi\mathop{=}\Quant x_1 \ldots \Quant x_n \ \psi\) where
\(x_i\mathop{\in}\VarTrace\mathop{\cup}\Prop\), for \(1\mathop{\leq} i\mathop{\leq}n\), and \(\Quant\mathop{\in}\{\forall, \exists\}\):
\((\traceAssign,{\setTraces}, k)\modelsHyper \varphi\) iff
\((\overline{\setTraces[k\ldots]},(\traceAssign_k, \AssignN^{\emptyset}))\modelsTFOL \transl_{H}(\varphi)\), where \(\AssignN^{\emptyset}\) is the empty time assignment.

We proceed by induction on the structure of the formula, with the base case being the (longest) quantifier-free formula, \(\psi\mathop{\in}\LTL\).
By Proposition~\ref{prop:zipping} and equivalence of \(\LTL\) to \(\FOLOrder\)~\cite{gabbay1980temporal}, \((\traceAssign,{\setTraces}, k)\modelsHyper \psi\) iff
\(\flatT{\traceAssign}[k \ldots ]\models \mathtt{LTLtoFO}(\psi)\).
Let \(\traceAssign'_k\) be the trace assignment defined from \(\flatT{\traceAssign}[k \ldots ]\) as follows: 
\(
a\mathop{\in}\traceAssign'_k(\pi)[i] \tIff
a_{\traceVar}\mathop{\in}\flatT{\traceAssign}[i+k]\)
and
\(q\mathop{\in}\traceAssign'_k(\traceVar_q)[i] \tIff
q\mathop{\in}\flatT{\traceAssign}[i+k]\)
for all \(\timeVar \mathop{\in}\nat\), trace variables \(\traceVar\in \Var\) and propositional variables \(a,q\mathop{\in} \Prop\).
We prove by induction on \(\FOLOrder\) formulas that 
\((\overline{\flatT{\traceAssign}[k \ldots ]}, \AssignN)\models \mathtt{LTLtoFO}(\psi)\) iff 
\((\overline{\setTraces[k\ldots]},(\traceAssign'_k, \AssignN))\modelsTFOL (\mathtt{LTLtoFO}(\psi))[\sigma']\), for any trace assignments \(\traceAssign\), time assignments \(\AssignN\), set of traces \(\setTraces\) and \(\LTL\) formula \(\psi\).
We observe that \(\sigma'\) changes only predicates, thus, we need only to prove the base cases, as the induction cases follow from definitions and induction hypothesis.
For the first base case,
\((\overline{\flatT{\traceAssign}[k \ldots ]}, \AssignN)\models \UnaryP_{a_{\traceVar}}(i)\) iff
\(a_{\traceVar}\mathop{\in}(\flatT{\traceAssign}[k \ldots ])[\AssignN(i)]\)
iff
\(a_{\traceVar}\mathop{\in}\flatT{\traceAssign}[k+\AssignN(i)]\).
By definition of \(\traceAssign'_k\), this is equivalent to 
\(a\mathop{\in}\traceAssign'_k(\traceVar)[\AssignN(i)]\), and so, 
\((\overline{\setTraces[k\ldots]},(\traceAssign'_k, \AssignN))\modelsTFOL \BinaryP_a(\traceVar,i)\).
The other base case, \((\overline{\flatT{\traceAssign}[k \ldots ]}, \AssignN)\models \UnaryP_{q}(i)\) is analogous.

Let's consider now the induction case, \(\exists \traceVar\ \varphi\).
We assume that \((\traceAssign,{\setTraces}, k)\modelsHyper \exists \traceVar\ \varphi\), that is, exists \(\trace\mathop{\in}\setTraces\) s.t.\ \((\traceAssign[\traceVar\mapsto\trace]{\setTraces}, k)\modelsHyper \varphi\).
By induction hypothesis, \((\overline{\setTraces[k\ldots]},(\traceAssign[\traceVar\mapsto\trace]_k, \AssignN^{\emptyset}))\modelsTFOL \varphi\)
and so
\((\overline{\setTraces[k\ldots]},(\traceAssign_k, \AssignN^{\emptyset}))\modelsTFOL \constrQ{\exists}{\varphi} \varphi\).
The other quantifier is analogous.
Hence,
for all \(\HQPTL\) formulas \(\varphi_H\) there exists the trace-prefixed hypertrace formula
\(\text{tr}_{H}(\varphi_H)\) s.t.\ \(\setTraces \modelsHyper \varphi_H\) iff \(\setTraces \modelsTFOL\text{tr}_{H}(\varphi_H)\).

For the translation from trace-prefixed hypertrace formulas to \(\HQPTL\), we use Lemma~\ref{lemma:rewrite:independent}, and give a translation from flattened hypertrace formulas to \(\HQPTL\).
As before we define a substitution 
\(\sigma''\mathop{=}\{\BinaryP_a(\traceVar, \timeVar) \mapsto P_{a_{\traceVar}}(i) \mid a\mathop{\in}\Prop,\traceVar\mathop{\in}\VarTrace, \timeVar\mathop{\in}\VarTime\}\cup
\{\BinaryP_q(\traceVar_q, \timeVar) \mapsto  P_{q}(i) \mid q\mathop{\in}\Prop,\traceVar\mathop{\in}\VarTrace, \timeVar\mathop{\in}\VarTime\}\).
The translation is defined as:
%\begin{align}
%\label{thm:tr_hypertrace_hqptl}
%\hspace{-4mm}
\(\transl_{\traceS}(\varphi) \mathop{=}
\overrightarrow{\Quant} (\mathtt{FOtoLTL}(\psi'[\sigma'']))\) with 
\(\overrightarrow{\Quant} \psi'\mathop{=}\rewriteFlat(\varphi,\Prop,\emptyset)
\)
where \(\overrightarrow{\Quant}\) is a sequence of constrained and unconstrained trace quantifiers and \(\psi'\) is a hypertrace formula with no trace quantifiers.

We prove that for all hypertrace formulas \(\varphi\) with only free trace variables,
for all sets of traces \(\setTraces\) and assignments \(\AssignT\):
\((\overline{\setTraces},(\AssignT, \AssignN^{\emptyset}))\modelsTFOL \varphi\) iff
\((\AssignT',{\setTraces}, 0)\modelsHyper \transl_{\traceS}(\varphi)\), where
\(\AssignT'(\traceVar_q)=(\AssignT(\traceVar_q)[0]\mathop{\cap}\{q\}) (\AssignT(\traceVar_q)[1]\mathop{\cap}\{q\}) \ldots\).
Base-case, where \(\psi'\) has no trace quantifiers and all time quantifiers are bounded: 
\((\overline{\setTraces},(\AssignT, \AssignN^{\emptyset}))\modelsTFOL \psi'\) iff
\((\overline{\flatT{\AssignT}}, \AssignN^\emptyset) \models \psi'[\sigma']\), follows from an analogous proof from the previous case.
As the translation from \(\AssignT\) to \(\AssignT'\) does not change the valuation of \(q\) in the trace assigned to \(\traceVar_q\) and 
by \cite{gabbay1980temporal},
\((\overline{\flatT{\AssignT}}, \AssignN^\emptyset) \models \psi'[\sigma']\) iff
\(\flatT{\AssignT'}\modelsLTL\mathtt{FOtoLTL}(\psi'[\sigma'])\).
By Proposition \ref{prop:zipping}, it is equivalent to
\((\AssignT',{\setTraces}, 0)\modelsHyper \mathtt{FOtoLTL}(\psi'[\sigma'])\).
The induction cases (i.e., constrained and unconstrained quantifiers) follow from induction hypothesis and definitions.
\end{proof}

\end{document}